\documentclass[journal,10pt]{IEEEtran}
\usepackage{cite}
\usepackage{graphicx}
\usepackage{amsmath}
\usepackage{times}
\usepackage{latexsym}
\usepackage{graphicx}
\usepackage{bm}
\usepackage{amssymb}
\usepackage[center]{caption2}
\usepackage{stfloats}
\usepackage{cases}
\usepackage{array}
\usepackage{setspace}
\usepackage{fancyhdr}
\usepackage{fancyhdr}
\usepackage{algorithm}
\usepackage{algpseudocode}
\usepackage{textcomp}
\usepackage{wasysym}
\usepackage{balance} 
\usepackage{flushend} 
\usepackage{url}
\usepackage{xcolor}

\allowdisplaybreaks[4]

\newcaptionstyle{mystyle1}{%
  \centering TABLE  \captiontext \par}
\captionstyle{mystyle1}
\newcaptionstyle{mystyle2}{%
  \captionlabel $.$ \: \captiontext \par}
\captionstyle{mystyle2}
\newcaptionstyle{mystyle3}{%
  \captionlabel $.$ \: \captiontext \par}
\captionstyle{mystyle3}

\renewcommand{\QED}{\QEDopen}

\newtheorem{theorem}{Theorem}

\begin{document}

\title{User Preference Learning Based Edge Caching for Fog Radio Access Network} 

\author{Yanxiang~Jiang,~\IEEEmembership{Senior Member,~IEEE},
Miaoli~Ma,
Mehdi~Bennis,~\IEEEmembership{Senior~Member,~IEEE}, Fu-Chun~Zheng,~\IEEEmembership{Senior~Member,~IEEE}, and~Xiaohu~You,~\IEEEmembership{Fellow,~IEEE}
\thanks{Manuscript received April 15, 2018, revised September 18, 2018, and accepted November 2, 2018. Part of this work has been presented at the 6th IEEE GLOBECOM Workshop on Emerging Technologies for 5G and Beyond Wireless and Mobile Networks (ET5GB), Singapore, December, 2017.
The associate editor coordinating the review of this paper and approving it for publication was T. He.}
\thanks{This work was supported in part by
the Natural Science Foundation of Jiangsu Province under grant
BK20181264,
the Research Fund of the State Key Laboratory of
Integrated Services Networks (Xidian University) under grant ISN19-10,
the Research Fund of the Key Laboratory of Wireless Sensor Network $\&$ Communication (Shanghai Institute of Microsystem and Information Technology, Chinese Academy of Sciences) under grant 2017002,
the National Basic Research Program of China
(973 Program) under grant 2012CB316004,
and the U.K. Engineering and Physical Sciences Research Council under Grant EP/K040685/2.}
\thanks{Y. Jiang is with the National Mobile Communications Research Laboratory, Southeast University, Nanjing 210096, China,
the State Key Laboratory of Integrated Services Networks, Xidian University, Xi'an 710071, China, and the Key Laboratory of Wireless Sensor Network $\&$ Communication, Shanghai Institute of Microsystem and Information Technology,
Chinese Academy of Sciences, 865 Changning Road, Shanghai 200050, China (e-mail: yxjiang@seu.edu.cn).}
\thanks{M. Ma and X. You are with the National Mobile Communications Research Laboratory, Southeast University, Nanjing 210096, China (e-mail: milies10@126.com, xhyu@seu.edu.cn).}
\thanks{M. Bennis is with the Centre for Wireless Communications, University of
Oulu, Oulu 90014, Finland (e-mail: mehdi.bennis@oulu.fi).}
\thanks{F. Zheng is with the School of Electronic and Information Engineering, Harbin Institute of Technology, Shenzhen 518055, China,
and the National Mobile Communications Research Laboratory, Southeast University, Nanjing 210096, China. (e-mail: fzheng@ieee.org).}
}

\maketitle

\begin{abstract}
In this paper, the edge caching problem in fog radio access network (F-RAN) is investigated. By maximizing the overall cache hit rate, the edge caching optimization problem is formulated to find the optimal policy. Content popularity in terms of time and space is considered from the perspective of regional users. We propose an online content popularity prediction algorithm by leveraging the content features and user preferences, and an offline user preference learning algorithm by using the {online gradient descent} (OGD) method and the {follow the (proximally) regularized leader} (FTRL-Proximal) method. Our proposed edge caching policy not only can promptly predict the future content popularity in an online fashion with low complexity, {but also} can track the content popularity with spatial and temporal popularity dynamic in time without delay.
{Furthermore, we design two learning based edge caching architectures.
Moreover, we} theoretically derive the upper bound of the popularity prediction error, the lower bound of the cache hit rate, and the regret bound of the overall cache hit rate of our proposed edge caching policy. Simulation results show that the overall cache hit rate of our proposed policy is superior to those of the traditional policies and asymptotically approaches the optimal performance.
\end{abstract}


\section{Introduction}

With the continuous and rapid proliferation of various intelligent devices and advanced mobile application services, wireless networks have been suffering from an unprecedented data traffic surge in recent years.
Although centralized cloud caching and computing in cloud radio access {network} (C-RAN) 
can provide reliable and stable service for end users during off-peak periods \cite{C-RAN, M-Bennis}, ever-increasing mobile data traffic brings tremendous pressure on C-RAN 
with capacity-limited fronthaul links and centralized baseband unit (BBU) pool, which may cause communication interruptions or traffic congestions especially at peak traffic moments. The main reason is that as {various} social applications become more and more popular, data traffic  over  fronthaul links surges with a lot of redundant and repeated information, which further worsens the fronthaul constraints. In this case, a feasible solution is to shift a small amount of resources such as communications, computing, and caching to network edge, and enable most of the frequently requested contents being served locally.
At this point,
fog radio access {network} (F-RAN) as a complementary network architecture was proposed, which can effectively reduce fronthaul load by placing most popular contents
closer to the requesting users
and extending  traditional cloud computing paradigm to  the network edge \cite{Bigdatacaching,F-RAN, Bennis}.
Up until now, F-RAN has  attracted more and more attention from researchers and engineers.
In F-RAN, traditional access points are turned into fog access points (F-APs) equipped with limited caching and computing resources, where edge caching is a key component to improve the performance of F-RAN. Due to storage constraints and fluctuating spatio-temporal traffic demands, however, F-APs face a myriad of challenges.
For example, how, what and when to strategically store contents in their local caches in order to achieve a higher cache hit rate?

Traditional caching policies such as first in first out (FIFO) \cite{FIFO}, least recently used (LRU) \cite{LRU}, least frequently used (LFU) \cite{LFU} and their variants \cite{ref_a} have been widely applied in wired networks, where there are abundant caching and computing resources and the served area is usually very large.
However,
these traditional caching policies may not be  applied well in wireless networks due to the characteristics of edge nodes such as smaller coverage areas and limited caching resources, and they
may suffer major performance degradation since they may not be able  to predict future content popularity correctly.
Most of the existing works  on edge caching in wireless networks assumed that the content popularity  is already known or subject to a uniform distribution, and then focused on exploring the optimal caching policy under {the above assumptions} \cite{games,factory,multicast}.
The edge caching problem was formulated as a many-to-many matching game  between small base stations and service providers' servers in \cite{games},  it {was} converted to an approximation facility location problem in \cite{factory}, and {the}  successful transmission probability was maximized   to obtain a local optimal caching and multicasting design in a general region in \cite{multicast}. However, these studies are inconsistent with the reality. Generally, content popularity is not as traceable as a unified distribution no matter what kind of caching policy is applied.

By considering the time-varying nature of contents, recent works have  turned to exploring sophisticated edge caching policies by learning future content popularity. In \cite{TL}, the content popularity matrix was estimated through transfer learning by leveraging user-content correlation and information transfer between time periods.
In \cite{Bharath}, the training time for transfer learning was analyzed for obtaining a better estimation performance.
Nevertheless, the content popularity matrix remains typically to be large, which needs a great deal of calculation in the estimation.
Meanwhile, the transfer learning approach has a poor performance for the case of low information correspondence ratio.
In \cite{CMAB}, the cache content placement problem was modeled as a contextual multi-arm bandit problem and an online policy was presented to learn the content popularity.
This policy  learns the content popularity independently across contents whereas
ignores the content similarity and {the} impact of user preference on content popularity, thereby resulting in high training complexity and slow learning speed.
A low complexity online policy was proposed in \cite{Trend}, where content popularity was learned based on the assumption that the expected popularities of similar contents are similar.
It performs  well for video caching but may be ineffective for other types of content caching.
However, all the above studies assume that the content popularity remains unchanged for a certain time period and the content library is stationary, and ignore to consider  spatial changes of content popularity, and  therefore cannot truly reflect the changes of content popularity.
In real communications scenarios, both the coverage area of an edge node and the number of users that it can  serve are limited.
The change of content popularity in the time and space dimensions is real-time.
{Due to the continuous emergence of new contents, the content library in the cloud content center in the time dimension must change.
Moreover, the set of  current users  in a specific region or scenario may change dynamically over time due to  user mobility, and the  content popularity may thus fluctuate too. Furthermore, due to the randomness of user requests, the content popularity will be dynamically changing over time.
In practice, different users may have different degrees of interest, i.e., different user preferences, in the same content. Correspondingly, the request possibilities for the same content among different regions or scenarios gathered by users with different user preferences are different. This will result in a  content popularity difference for the same content among different regions or scenarios.}
These changing factors make it impossible to measure the content popularity merely through a unified distribution or a simple prediction. The small changes of the content popularity directly affect the caching decisions. Therefore, it is necessary to explore spatial and temporal dynamic of  content popularity and track the dynamic in a timely manner for the ensurance of continuously caching {popular} contents, the achievement of optimal caching decisions, and the maximization of cache hit rate.

Motivated by the aforementioned discussions, our main contributions are summarized below.
\begin{enumerate}
 \item We put forward a new idea of content popularity prediction.
  Unlike the static approach, we consider content popularity in terms of time and space from the perspective of regional users
  and propose an online content popularity prediction algorithm, which can predict the future content popularity of a certain region in an online fashion without any restriction on content types 
and track the popularity change in real time.
   \item We propose an offline user preference learning algorithm, which can discover the user's own preference through its historically requested information.
      By monitoring the average prediction error in real time, it can be initiated automatically for relearning of user preference and continuous offline learning can thus be avoided.
   \item We design two learning based edge caching architectures for F-RAN.
    They differ in  that  the offline user preference learning functionality is transferred from the F-APs in the first architecture to the  smart user equipments (UEs) in the second one.
    Specifically, we introduce cooperative caching among regional F-APs and make more efficient usage of limited caching resources. Besides, we introduce a new module to enable regular monitoring of regional users by considering the impact of user mobility on cache decisions.
   \item We analyze the performance of our proposed edge caching policy. We first derive the upper bound of the popularity prediction error of our proposed online content popularity prediction algorithm, and reveal the sub-linear relationship between the cumulative prediction error and the total number of content requests.
       We then derive the regret bound of the overall cache hit rate of our proposed edge caching policy, and show through theoretical analysis that our proposed  policy has the capability to achieve
the optimal performance asymptotically.
    \item We validate our theoretical results with real data. Simulation results show that  our proposed edge caching policy can predict the content popularity with high precision, and track the real local popularity {in real time} with spatial and temporal popularity dynamic. Simulation results also
    show the superior performance of our proposed  policy in comparison with the other traditional policies and verify its asymptotical optimality.
\end{enumerate}

The rest of this paper is organized as follows. In section II, the system
model
is described.
Our proposed edge caching policy including  online content popularity prediction  and offline user preference learning is presented in Section III.
{The two learning based edge caching architectures  are described in Section IV.
The performance analysis is provided in Section V.}
Simulation results are shown in Section VI.
Final conclusions are drawn in Section VII.

\section{System Model}

\subsection{Edge Caching Scenarios in F-RAN}

We consider the edge caching scenarios in F-RAN as illustrated in Fig. \ref{fig1}.
Large amounts of F-APs with limited storage capacity are deployed in different scenarios, for example, a town with  dense crowds and moderate user mobility, a park with relatively dense crowds and high user mobility,
or a stadium with ultra-dense crowds and relatively low user mobility.
Due to different user density and user mobility in different scenarios, there exist differences in the distribution of {popular} contents among different scenarios. In order to ensure a stronger scenario suitability of the caching policy,
we propose the following edge caching design rule:
The F-APs deployed in different scenarios should regularly monitor the users in consideration of user mobility, and set a  monitoring cycle matched with the characteristics of different scenarios.

\begin{figure}[!t]
\centering 
\includegraphics[width=0.36 \textwidth]{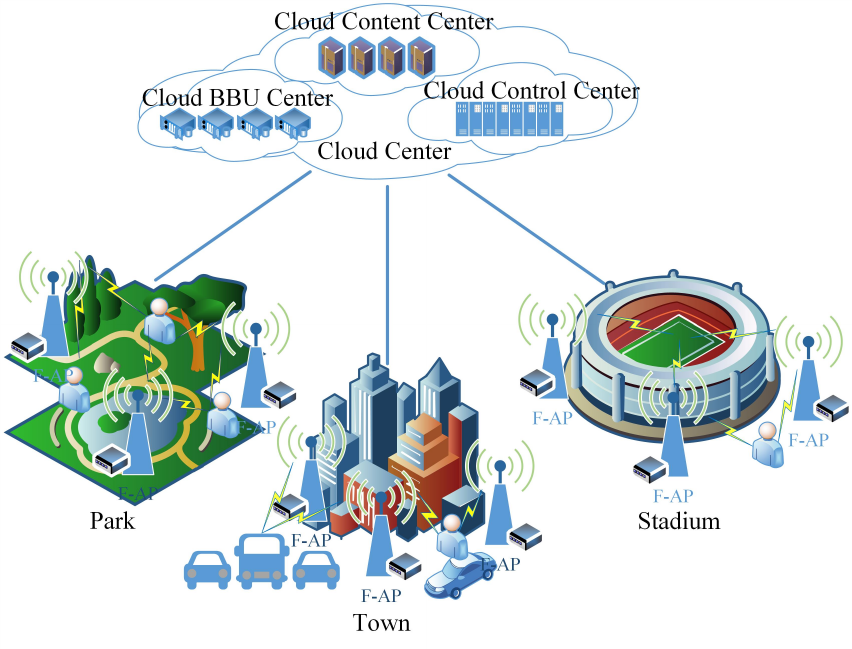}
\caption{Illustration of the edge cache scenarios in F-RAN.}
\label{fig1}
\end{figure}

In each edge caching scenario, according to the corresponding design criteria (for example, location), the F-APs with close distance can cooperate and belong to the same region, which are also called regional F-APs.
The users in the same region, also called regional users, can request  contents of interest from
their associated regional F-APs if the contents are stored in the  local caches
or from neighboring F-APs in the other regions or the cloud content center through  fronthaul links otherwise. 
Then, we {focus} on the caching policy for a single region, namely the regional caching.

\subsection{Edge Caching Problem Formulation}
We consider the edge caching problem in a specific region served by ${M}$ F-APs, which  constitute an F-AP set ${\cal M} = \left\{ {1,2,\cdots,m,\cdots,M} \right\}$.
It is assumed that the regional users  can  fetch the contents of interest from the  local caches of the $ M$ F-APs.
Without loss of generality, we assume that all the  contents have the same size\footnote{{Note that contents with different sizes can always be split into data segments of the same size, and each data segment can then be considered as a ``content". This is a common practice in real world systems. Here we follow the same assumption and justification as in \cite{CMAB} and \cite{SLi}.}} and each F-AP has the same storage capacity and can store up to $\varphi $ contents from the content library ${\cal F} = \left\{ {1,2,\cdots,f,\cdots,F} \right\}$,
which may vary over time and is located in the cloud content center.
Without loss of generality, assume {$M \varphi \ll F $}.
By considering user mobility
and according to the preset monitoring period,
the F-APs
monitor the users in the specific region 
during discrete time periods $t = 1,2,\cdots,T$, where $T$ is set to be a finite time horizon.
Let ${U_t}$  denote {the number of regional users} served by the  ${M}$ F-APs during the $t$th time period with {${U_{\min }} \le {U_t} \le {U_{\max }}$,}
and ${{\cal U}_t} = \left\{ {1,2,\cdots,{U_t}} \right\}$  the  set of regional users monitored by the  ${M}$ F-APs.
It is assumed that the regional users remain unchanged during the considered time period. 
In the way of cache content placement, for description convenience, we adopt the partition-based content placement method in  \cite{copCaching}, where
each content is separated into ${M}$ equal-sized subfiles, and the F-APs 
store its different subfiles.\footnote{{Note that it does incur the cost for transferring data among the regional F-APs.}}
\footnote{{The $M$ regional F-APs can indeed be treated as a single cache by adopting the simple partition-based content placement method.
In this regard, the intelligent caching policies on single cache
such as \cite{TL,Bharath,CMAB,Trend} can be applied here, where transfer learning, contextual multi-arm bandit, and trend-aware online learning were employed, respectively. However, just as stated previously, these intelligent caching policies have made strict assumptions and ignored to consider spatial dynamic of content popularity.}}
\footnote{{On the other hand, however, the considered
$M$ regional F-APs are not simply treated as a single cache. In comparison with the single-cache setting, the cooperative-cache setting here has the following advantages: 1) By separating each content into $M$ equal-sized subfiles and exploiting cooperation, the transmission delay can be decreased. Specifically, for video content, users may just want to watch part of it. Correspondingly, only some of the subfiles will need to be transmitted. Therefore, the transmission delay can be further decreased.
2) Furthermore, the coverage area of multiple F-APs and the number of their served users are relatively large.  Correspondingly, popular contents can be more ``concentrated" (i.e. in the the coverage area). Therefore, a higher cache hit rate can be achieved.
3) Besides, more contents can be cached in the cooperative F-APs in a specific region, and unnecessary caching redundancy can be avoided.}}
We remark here that more sophisticated  content placement methods can be adopted.


Let ${D_t}$ denote {the number of requests} during the $t$th time period,  $ D = \sum\nolimits_{t = 1}^T {{D_t}} $ the overall number of requests in the finite time horizon ${T}$,
and $ {req}_t = \left\{ {req_{t,1},re{q_{t,2}},\cdots,re{q_{t,d}},\cdots,re{q_{t,{D_t}}}} \right\} $ the set of requests which come in sequence during the $t$th time period.
The request $re{q_{t,d}}$ can be further expressed as
$re{q_{t,d}} = \left\langle {f\left( d \right),t\left( d \right), \boldsymbol x\left( d \right)} \right\rangle ,
\forall 1 \le t \le {T}, \forall 1 \le d \le {D_t},$
where $ f\left( d \right) \in {\cal F}$ denotes the requested content, $ t\left( d \right)$ the requesting time, {and} $ {\boldsymbol{x}}\left( d \right) \in {\mathbb R^N} $ the feature vector with dimension $N$
describing the features of the requested content.
Take movie as an example: $ {\boldsymbol{x}}\left( d \right)$ may include features like the  movie rating, the movie type, the keyword frequency of movie critics,
etc.
Without loss of generality, we normalize the various dimensions of ${\boldsymbol{x}(d)} $ and set ${\boldsymbol{x}(d)} \in {\left[ {0,1} \right]^N} $. 

During the $t$th time period, for each arriving request $re{q_{t,d}}$, the regional F-APs first check whether
$f(d)$ has been stored locally.
Let ${\theta _{t,d}}\left( {f\left( d \right)} \right) \in \left\{ {0,1} \right\}$ denote the cache status of  $ f\left( d \right)$ at the requesting time $t(d)$,
where ${\theta _{t,d}}\left( {f\left( d \right)} \right) = 1$ if ${f\left( d \right)}$ is stored locally, 
and ${\theta _{t,d}}\left( {f\left( d \right)} \right) = 0$ otherwise. If  ${f\left( d \right)}$ has been stored in the local caches, a cache hit happens and the requesting user can then be served locally.
Otherwise, a cache miss happens and ${f\left( d \right)}$ will be fetched from neighboring F-APs in the other regions or the cloud content center, and a caching decision will be further made to determine whether to store ${f\left( d \right)}$ locally.
If the F-APs decide to store $f(d)$ and replace one of the existing contents in the local caches, denoted by ${f_{\rm old}} \in \left\{ {f | {\theta _{t,d}}\left( f \right) = 1, \forall f \in {\cal F}} \right\}$, a cache update happens.
Then, the cache status will be updated 
according to the following rule
\begin{align}
{\theta _{t,d + 1}}\left( f \right) = \left\{ {\begin{array}{*{20}{l}}
0,&{{\text{if}} \ f = {f_{{\rm{old}}}}},\\
1,&{{\text{else if}} \ f = f\left( d \right)},\\
{{\theta _{t,d}}\left( f \right)},&{{\text{else if}} \ f \in {\cal F}\backslash \left\{ {f\left( d \right),{f_{{\rm{old}}}}} \right\}}.
\end{array}} \right.
\end{align}
In addition,
{a caching decision that $f(d)$ will not be stored may be made, and then the cache status will remain unchanged,}
i.e., ${\theta _{t,d + 1}}\left( f \right) = {\theta _{t,d}}\left( f \right),\forall f \in {\cal F}$.

For convenience, we use ${{\boldsymbol{\theta }}_{t,d}} = \left[ {\theta _{t,d}}\left( 1 \right),  {\theta _{t,d}}\left( 2 \right),\cdots,{\theta _{t,d}}\left( f \right),\cdots, {\theta _{t,d}}\left( F \right) \right]^{{T}}$ to indicate the cache status of all the contents
at the requesting time $t(d)$.
Generally, an edge caching policy can be represented by a function ${\Phi}:\left(  {{\boldsymbol{\theta }}_{t,d}} , {\boldsymbol{x}}\left( d \right) , {{\cal U}_t} \right) \mapsto {{{\boldsymbol{\theta }}_{t,d + 1}}} $, which maps the current cache status vector, the current feature vector, and the  set of regional users  to the next cache status vector.
After a request ${re{q_{t,d}}}$ is served, the cache status vector should be updated according to the edge caching policy.
We use the overall cache hit rate to evaluate the caching performance,
which is defined as the number of cache hits over the whole requests during the finite time horizon $T$ as follows
\begin{align}\label{eq_0}
\mathcal H \left( \Phi  \right)
&= \frac{{\sum\nolimits_{t = 1}^T {\sum\nolimits_{d = 1}^{{D_t}} {{\theta _{t,d}}\left( {f\left( d \right)} \right)} } }}{{\sum\nolimits_{t = 1}^T {{D_t}} }} \nonumber\\
&= \frac{1}{D}\sum\limits_{t = 1}^T {\sum\limits_{d = 1}^{{D_t}} {{\theta _{t,d}}\left( {f\left( d \right)} \right)} } .
\end{align}
Then,
the corresponding edge caching optimization problem\footnote{{Note here that the focus of this paper is content popularity prediction and user preference learning, which is one of the most important parts of edge caching or F-RAN.
To highlight this issue, we simplify the content placement and adopt the partition-based content placement method and use the cache hit rate as the optimization objective, which help us concentrate on the investigation of content popularity prediction and user preference learning.
If more sophisticated content placement methods and other optimization objectives are considered, the corresponding optimization problem will have more F-RAN-specific parameters or communication-related parameters.
Actually, the investigation results concerning content popularity prediction in this paper can be readily extended to deal with non-simplified scenarios such as \cite{Jiang-icnc18, Jiang-vtc18fall, Jiang-globecom18} by using the predicted content popularity rather than the assumed Zipf distribution.}} can be expressed as follows \cite{Jiang-conf}
\begin{align}\label{eq_1}
&{\mathop {\max }\limits_\Phi  }{\mathcal H  \left( \Phi  \right)},\\
 {{\rm{s}}.{\rm{t}}.\ } & {{\theta _{t,d}}\left( f \right) \in \left\{ {0,1} \right\}, \text{for} \ 1 \le d \le {D_t}, 1 \le t \le T, \forall f \in \cal F,} \nonumber \\
{}&    {\boldsymbol{\theta }}_{t,d}^T{{\boldsymbol{\theta }}_{t,d}} \le M\varphi ,  \text{for} \ 1 \le d \le {D_t}, 1 \le t \le T. \nonumber
\end{align}

Our objective in this paper is to find the optimal edge caching policy by maximizing the overall cache hit rate over the finite time horizon $T$ with the limited total cache size $M\varphi$.
For convenience, a summary of major notations
is presented in Table \ref{table1}.

\begin{table*}[!t]
\small
\caption{Summary of major notations.}\label{table1}
\begin{tabular}{|c|m{14cm}|}
\hline
\hline
$M$, ${\cal M}$ & Number of regional F-APs, set of the $M$  regional F-APs \\
\hline
${t}$, ${T}$ & Discrete time periods,  finite time horizon \\
\hline
$\varphi $ & Cache size of  each F-AP \\
\hline
${U_t}$, ${U_{\max }}$, ${U_{\min }}$, ${{\cal U}_t}$ & Number of regional users during the $t$th time period,
  maximum ${U_t}$,  minimum ${U_t}$, Set of the ${U_t}$ regional users \\
\hline
${D_t}$,  ${D}$ & Number of   requests during the $t$th time period,  number of overall requests in the finite time horizon $T$ \\
\hline
$re{q_t}$, $re{q_{t,d}}$ & Set of requests coming in sequence during the $t$th time period, the $d$th request during the $t$th time period\\
\hline
$f$, $f\left( d \right)$,  $f_{\rm{old}}$, {${f^{{\rm{smallest}}}}$} & Content, the $d$th requested content,  content to be removed, {content with the smallest popularity in the current local cache} \\
\hline
$t\left( d \right)$, {${t_f}$, ${t_{f^{\text{smallest}}}}$} & The $d$th requesting time, {initial caching time of the content $f $, initial caching time of the content $f^{\text{smallest}}$}  \\
\hline
{${{\cal G}_{t,d}}$, $P_{t,f}^{{\rm{cur}}}$, ${P^{{\rm{smallest}}}}$} & {Set of current contents in the local cache at the requesting time $t\left( d \right)$, current popularity of the caching content $f \in {{\cal G}_{t,d}}$  after it is requested, the smallest content popularity} \\
\hline
{${Q_{t,d}}$} & {Priority queue that stores the caching contents and information sequentially} \\
\hline
${\cal F}$, $F$ & Content library, size of content library \\
\hline
${\boldsymbol{x}}\left( d \right)$, {${{\boldsymbol{x}}^{\left( k \right)}}$,} $N $& Feature vector of the $d$th requested content, {feature vector of the $k$th sample,} dimension of feature vector\\
\hline
${\theta _{t,d}}$, ${{\boldsymbol{\theta }}_{t,d}}$ & Cache status of the requested content $f(d)$ at the requesting time $t\left( d \right)$,
vector of  cache status of all the contents at the requesting time $t\left( d \right)$  \\
\hline
$\Phi $, {$\Phi ^*$} & Edge caching policy, {the optimal edge caching policy} \\
\hline
{${{\cal H}_t}\left( {{\Phi }}\right)$, ${{\cal H}_t}\left( {{\Phi ^*}}\right)$} & {Cache hit rate of $\Phi$ during the $t$th time period, the optimal cache hit rate during the $t$th time period}\\
\hline
 ${\cal H}\left( \Phi  \right)$, {${{\cal H}}\left( {{\Phi ^*}} \right)$} & Overall cache hit rate of  $\Phi $ in the finite time horizon $T$, {the optimal overall cache hit rate  in the finite time horizon $T$} \\
\hline
{${{\boldsymbol{w}}_u}$, ${\boldsymbol{w}}_u^{\left( k \right)}$, ${{\boldsymbol{w}}_u^*}$} & {Vector of user preference model parameters of the user $u$, vector of user preference model parameters of the user $u$ for the $k$th iteration, vector of the optimal user preference model parameters} \\
\hline
{${{\hat p}_{t,u,d}}$, ${{ p}_{t,u,d}}$}  & {Predicted possibility that the user  $u $ requests  ${f\left( d \right)}$ at $t(d)$ during the $t$th time period, real possibility that the user  $u $ requests  ${f\left( d \right)}$ at $t(d)$ during the $t$th time period} \\
\hline
{${{\hat P}_{t,d}}$,  $ P_{t,d}$, $\hat P_{t,d'}^{'}$, $P_{t,d'}^{'}$}   & {Predicted  popularity of $f\left( d \right)$ at $t\left( d \right)$,  real  popularity  of $f\left( d \right)$ at $t\left( d \right)$,  predicted   popularity with respect to $ P_{t,d'}^{'}$, popularity of the ${d'}$th most popular content when it is requested firstly during the $t$th time period}\\
\hline
{$y\left( d \right)$, ${y^{\left( k \right)}}$} & {Category label of the $d$th requested content, category label of the $k$th sample}\\
\hline
{${K_{t,u,d}}$, $K$} & {Cumulative number of samples for the user $u$ from the last model update to the time when the request $re{q_{t,d}}$ arrives,  number of collected samples for offline user preference learning} \\
\hline
{$\ell \left( {{\boldsymbol{w}_u},{\boldsymbol{x}(d)},y(d)} \right)$, ${\xi _{t,u,d}}$} & {Logistic loss  for the user $u$ at $t\left( d \right)$, average logistic loss for the user $u$ at $t\left( d \right)$} \\
\hline
{$\gamma $, ${{\eta ^{\left( k \right)}}}$, $\sigma^{(k')}$} & {Predefined threshold, non-increasing learning-rate schedule, parameter that has certain relationship with $\eta^{(k)}$} \\
\hline
{${{\boldsymbol{g}}^{\left( k \right)}}$, ${{\bf{g}}^{\left( {1:k} \right)}}$} & {Gradient vector of the logistic loss of the $k$th sample with
respect to ${{\boldsymbol{w}}_u}$, sum of the gradient vectors of the logistic loss of the previous $k$ samples}\\
\hline
{$\lambda_1$, $\lambda_2$; $\alpha, \beta$; $\tau_u$} & {Regularization parameters with positive values; adjusting parameters; sufficiently small constant with a positive value} \\
\hline
{$\mathcal{E} $, ${{L_d}}\left( {{{\boldsymbol{w}}_u}} \right)$} & {Convex set of the optimization problem in \eqref{eq11}, a sequence of convex loss functions} \\
\hline
{$F_t$} & {number of requested different contents during the $t$th time period} \\
\hline
{$R\left( D \right)$} & {Regret of the overall cache hit rate for the total $D$ requests in the finite time horizon $T$} \\
\hline
\hline
\end{tabular}
\end{table*}

\section{The Proposed User Preference Learning Based Edge Caching Policy}

In order to maximize the cache hit rate,
we propose a novel edge caching policy which includes
an online content popularity prediction algorithm and an offline user preference learning algorithm.
The proposed policy can  continuously cache popular contents based on the content features and  user preferences.\footnote{{Note that the adaptive caching scheme recently proposed in \cite{Tanzil} also considers content popularity prediction based on content features and users' behavior. However, it uses an extreme-learning machine (ELM) neural network to estimate the content popularity, and more addresses the optimization of the number of neurons, the construction and selection of content features. Besides, its definition of content popularity is also different with ours.}}

\subsection{Policy Description}

The detailed edge caching policy is shown in Algorithm 1.
The considered $M$  F-APs serving in the region set a fixed monitoring period
and periodically monitor the current user set in the region.
During the $t$th time period, the $M$ F-APs first obtain the current user set ${{\cal U}_t}$.
For each arriving request $re{q_{t,d}}$ from the user $u \in {\cal U}_t$,  its request information will then be recorded.
Meanwhile,  the features of  $f\left( d \right)$ are extracted and recorded. 
The recorded data will be used to train or update the user preference model in order to improve the prediction precision of the content popularity.
Let ${{\cal G}_{t,d}} = \left\{ {f | {\theta _{t,d}}\left( f \right) = 1, \forall f \in {\cal F}} \right\}$ denote the set of current contents in the local cache.
It will be explored to see whether $f\left( d \right)$ has already been stored locally.
Just as stated in the previous section,
the corresponding caching decision will be made to determine whether $f\left( d \right)$ should be cached and which stored content should be removed from the storage space of the $M$ F-APs when $f\left( d \right)$ needs to be cached.

\begin{algorithm}[!t]
\small
	\caption{The proposed edge caching policy}
	\label{alg:1}
	\begin{algorithmic}[1]
        \Procedure{EdgeCaching}{$req_{t,d}$}
		\For{$t = 1,2,\cdots,T$, }	
		\State The considered $M$ F-APs monitor  ${{\cal U}_t}$;
        \For{$d = 1,2,\cdots,D_t$, }
		\State Record the request information  of  $ re{q_{t,d}}$;
         Read the set  of the current caching content
         ${{\cal G}_{t,d}}$;
		\If {$f\left( d \right)$ has been stored locally,} 
		\State The users are served locally;
         $ P_{t,{f\left( d \right)}}^{\text{cur}} = P_{t,{f\left( d \right)}}^{\text{cur}} - {1 \mathord{\left/ {\vphantom {1 {{U_t}}}} \right. \kern-\nulldelimiterspace} {{U_t}}}$;
        \Else
        \State Fetch $f\left( d \right)$ from the cloud content center
         or the neighboring F-APs in the other regions;
         \State Extract ${\boldsymbol {x}}\left( d \right)$ and   $\left\{{\boldsymbol {w}}_u | \forall u \in {\cal U}_t \right\}$;
         ${\widehat P_{t,d}} \leftarrow {\rm{Predict}}\left({\boldsymbol{x}}(d), \left\{ {{{\boldsymbol{w}}_u}} | \forall u \in {\cal U}_t  \right\} \right) $;
        \State Sort ${Q_{t,d}}$ based on $ P_{t,f}^{\text{cur}}$ and  ${t_f}$ for $f \in {{\cal G}_{t,d}}$;
         Get {$P^{\text{smallest}}$, ${f^{\text{smallest}}}$} from the top  of ${Q_{t,d}}$;
         \If {${\widehat P_{t,d}} > {{ P^{\text{smallest}}}}$,}\Comment{Cache update}
        \State Remove the top element from ${Q_{t,d}}$;
         ${t_{f\left( d \right)}} = t\left( d \right)$; $ P_{t,{f\left( d \right)}}^{\text{cur}} = {\widehat P_{t,d}} - {1 \mathord{\left/ {\vphantom {1 {{U_t}}}} \right. \kern-\nulldelimiterspace} {{U_t}}}$;
        \State Insert $\langle { P_{t,{f\left( d \right)}}^{\text{cur}},f\left( d \right), t_{f\left( d \right)}} \rangle $ into ${Q_{t,d}}$;
         Replace {${f^{\text{smallest}}}$} by $f\left( d \right)$.
        \EndIf
        \EndIf
        \EndFor
        \EndFor
        \EndProcedure
	\end{algorithmic}
\end{algorithm}

In order to make an optimal caching decision, the feature vector ${\boldsymbol{x}}\left( d \right)$
and the vectors of {the} well-trained user preference model parameters of all the users  $\left\{ {{\boldsymbol{w}}_u} | \forall u \in {\cal U}_t \right\}$ are extracted 
to predict the popularity ${ P_{t,d}}$ of the requested content $f\left( d \right)$.
In addition, considering that the content popularity will change over time, in order to track popularity changes,
we let $  P_{t,f}^{\text{cur}}$  denote  the current popularity of the caching content $f \in {{\cal G}_{t,d}}$  after it is requested  (also called the residual request rate).
We know that users may have a certain request delay on the same content.
In order to ensure timely and reasonable {cache update}, we propose to select the content with the characteristics of the smallest content popularity {${ P^{\text{smallest}}}$} and relatively earlier initial cache time {${t_{f^{\text{smallest}}}}$},\footnote{{Note that here we propose to select the content with the smallest content popularity in the local caches. Furthermore, if there exist multiple contents that have the same smallest content popularity in the local caches, we propose to select the content with the smallest content popularity that has the longest cache time, i.e., the content that has been cached earliest.}} denoted by {${f^{\text{smallest}}}$}, as the content to be removed from the local caches. 
In order to  locate {${f^{\text{smallest}}}$} quickly, we propose to reserve a priority queue ${Q_{t,d}}$ that stores the caching contents along with their current content popularity $ P_{t,f}^{\text{cur}}$ and their initial caching time ${t_f}$ for $f \in {{\cal G}_{t,d}}$.
The elements of ${Q_{t,d}}$ are sorted
in sequence when the request $req_{t,d}$ arrives, whose top element is composed of 
{${f^{\text{smallest}}}$, 
${t_{f^{\text{smallest}}}}$} and 
{${ P^{\text{smallest}}}$.}
A caching decision is made by comparing the predictive popularity  ${\widehat P_{t,d}}$ and {${ P^{\text{smallest}}}$.}
If ${\widehat P_{t,d}}$ is larger than {${ P^{\text{smallest}}}$,} the existing content {${f^{\text{smallest}}}$} will be  replaced by $f\left( d \right)$, the initial caching time of $f\left( d \right)$ will be recorded, and the current popularity of  $f\left( d \right)$ and the priority queue ${Q_{t,d}}$ will be updated  accordingly. After that, a cache update process is completed. Otherwise, nothing will be done to the local caches.
{The key here, obviously, is to obtain $\widehat P_{t,d}$, which will be described in the next subsection.}

\vspace*{-10pt}
\subsection{Online Content Popularity Prediction}

In this subsection, we  propose an online content popularity prediction algorithm based on the content features and user preferences.
During the $t$th time period, for each requesting user $u \in {\cal U}_t$, the requested content can be classified into a {preferred} category or an {unpreferred}  one for this user based on its user preference.
Generally,  a user  prefers to request  contents of its {preferred} category.
The problem of whether a user will request a certain content can be converted into a simple two-category one.
We  use the sigmoid function to approximate the correspondence between the  feature vector and the category label of the requested content \cite{sigmoid},  and construct a logistic regression model to approximate the user preference model.
For the arriving request $re{q_{t,d}}$, it is characterized by the feature vector ${\boldsymbol{x}}\left( d \right)$.
Let $y(d) $ denote the corresponding category label with $y(d)=1$ if the requested content is the {preferred} category of the user and $y(d)=0$ otherwise.
Let ${ p_{t,u,d}}$ denote the possibility that the user  $u \in {\cal U}_t$ requests the content ${f\left( d \right)}$ at the requesting time $t(d)$ during the $t$th time period.
Specifically, we assume that a user will not have a second request to the same content\footnote{Note that the above assumption is for all the contents and has a fair impact on the content popularity.
By using this assumption, neither the {popular} contents can become  unpopular ones nor the unpopular contents can become the {popular} ones.
Furthermore, in some ways, the content which a user has requested repeatedly can be obtained directly from this user's own cache and does not need to be obtained repeatedly from the corresponding F-APs.
}.
If the  user $u $ has already requested ${f\left( d \right)}$ previously,
then ${\widehat p_{t,u,d}} = 0$.
Otherwise, ${ p_{t,u,d}}$ can be predicted based on  the following user preference model
\begin{equation}
{{\widehat p}_{t,u,d}} = {p_{\boldsymbol w_u}}[ {y\left( d \right) = 1\left| {{\boldsymbol{x}}\left( d \right)} \right.} ]
= \frac{1} {{{1 + {e^{ - ( {{{\boldsymbol{w}}_u^T}*{\boldsymbol{x}}\left( d \right)} )}}}}}.
\end{equation}
Furthermore, the regional content popularity $ P_{t,d}$  can be  predicted by using  the average request possibility from the regional users for  $f\left( d \right)$ as follows
\begin{align}\label{eq5}
{\widehat P_{t,d}} = \frac{1}{{{U_t}}} {{\sum\limits_{u = 1}^{{U_t}} {{{\widehat p}_{t,u,d}}} }}.
\end{align}
After that, 
if the requested content is determined to be stored into the local cache or it has been stored locally and is requested again, according to our previous assumption that a user will not have a second request to the same content, its current popularity at the requesting time can be calculated respectively as  follows
\begin{equation}
P_{t,{f\left( d \right)}}^{\text{cur}} = {\widehat P_{t,d}} - {1}/{U_t}, \
P_{t,{f\left( d \right)}}^{\text{cur}}  = P_{t,{f\left( d \right)}}^{\text{cur}} - {1}/{U_t}.
\end{equation}
Besides, if the requested content is determined not to be stored locally and is requested by another regional user again, its current popularity at the next requesting time can be predicted by \eqref{eq5} since we have set the possibility that the previous requesting user will request this content again to zero. In this way, our proposed online content popularity algorithm can track the popularity changes in real time.

To measure the prediction performance, we introduce the logistic loss $\ell \left( {{\boldsymbol{w}_u},{\boldsymbol{x}(d)},y(d)} \right)$ for the user $u$, which is defined as the negative log-likelihood of $y(d)$ given ${p_{\boldsymbol{w}_u}}( {\left. y(d) \right|{\boldsymbol{x}(d)}} )$ and can be expressed as follows
\begin{multline}
\ell ( {{\boldsymbol{w}_u},{\boldsymbol{x}(d)},y(d)} )
= -y(d)\log {p_{\boldsymbol{w}_u}}({\left. y(d) \right|{\boldsymbol{x}(d)}} ) \\
- \left( {1 - y(d)} \right)\log [ {1 - {p_{\boldsymbol{w}_u}}( {\left. y(d) \right|{\boldsymbol{x}(d)}} )} ].
\end{multline}
Since user preference may change over time, we need to  capture the moment when the user preference changes and {update the user preference model in real time}. 
For this purpose, we collect the samples $\{({\boldsymbol{x}}^{(k)}, y^{(k)})\}_{k=1}^{K_{t,u,d}}$ and  monitor the prediction performance in real time,
where ${{{\boldsymbol{x}}^{\left( k \right)}}}$ and ${{y^{\left( k \right)}}}$
denote the feature vector and   category label of the $k$th sample, respectively,
and ${{K_{t,u,d}}}$   the cumulative number of  samples  for  the user $u$ from the last model update to the time when the request $re{q_{t,d}}$ arrives during the $t$th time  period. 
Then, the average logistic loss
for  the user $u $ 
can be expressed  as follows
\begin{equation}
{\xi _{t,u,d}} = \frac{1}{{{K_{t,u,d}}}}\sum\limits_{k = 1}^{{K_{t,u,d}}} {\ell ( {{{\boldsymbol{w}}_u},{{\boldsymbol{x}}^{\left( k \right)}},{y^{\left( k \right)}}} )}.
\end{equation}
Let $\gamma $ denote a predefined threshold with $0 \le \gamma  \le 1$. 
When ${\xi _{t,u,d}}$ exceeds $\gamma $, the update of the user preference model  will be initiated.

\begin{algorithm}[!t]
\small
	\caption{The online content popularity prediction algorithm}
	\label{alg:2}
	\begin{algorithmic}[1]
        \Procedure  {Predict}{${\boldsymbol{x}}(d), \left\{ {{{\boldsymbol{w}}_u}} | \forall u \in {\cal U}_t  \right\} $} 		\For{$u$ $\in {{\cal U}_t} $, }
        \If {The user $u$ has requested the content,} 
        \State ${\widehat p_{t,u,d}} = 0$;
        \Else
		\State Obtain the well-trained ${{\boldsymbol{w}}_u}$ from the $M$ F-APs;
		 ${\widehat p_{t,u,d}} = {1}/({1 + {e^{ - ( {{{\boldsymbol{w}}_u^T}{\boldsymbol{*x}}\left( d \right)} )}}})$;
        Observe the category label ${y(d)}$ of  ${f\left( d \right)}$.
        \EndIf
        \State Get  $( {{{\boldsymbol{x}}^{( {K_{t,u,d}} )}},{y^{({K_{t,u,d}})}}} )$;
         ${K_{t,u,d}} = {K_{t,u,d - 1}} + 1$;
         ${\xi _{t,u,d}} = {
        {{{\xi _{t,u,d - 1}} + \frac{1}{ K_{t,u,d}} \ell ( {{{\boldsymbol{w}}_u},{{\boldsymbol{x}}^{( {K_{t,u,d}} )}},{y^{( {K_{t,u,d}} )}}} )}}} $;

        \If {${\xi _{t,u,d}} \ge \gamma $, }
		\State $\boldsymbol{w}_u \leftarrow$  Learn$( {\{ {( {{{\boldsymbol{x}}^{\left( k \right)}},{y^{\left( k \right)}}} )} \}_{k = 1}^{{K_{t,u,d}}}} )$.
        \EndIf
        \EndFor
        \State \textbf{return} $\frac{1}{{{U_t}}} {{\sum\nolimits_{u = 1}^{{U_t}} {{{\widehat p}_{t,u,d}}} }}$.
        \EndProcedure
	\end{algorithmic}
\end{algorithm}

The detailed description of our proposed online content popularity prediction algorithm is shown in Algorithm 2.
Note that our proposed algorithm not only  can predict  content popularity in an  online fashion, {but also} can determine when to update the user preference model proactively.
We  also remark  here that the  time complexity of our proposed algorithm when making prediction for one content is ${\mathcal O}\left( U_t \right)$.

\subsection{Offline User Preference Learning}

With the rapid development of social networks and various multimedia applications, each user will access a large amount of contents everyday. In order to provide users with more intelligent services and enhance the quality of experience,
it is necessary and feasible to establish the independent user preference model for each user.
When the update of the user preference model  is initiated, we assume that there are $K$ samples
collected for the considered user and extracted from the recorded data,
which are denoted by $ \left\{ {\left( {{{\boldsymbol{x}}^{\left( k \right)}},{y^{\left( k \right)}}} \right)} \right\}_{k = 1}^K$.
Based on the collected samples, we propose to learn the 
user preference model parameters iteratively  by minimizing
the logistic loss of each sample as follows
\begin{equation}\label{eq-10}
{\boldsymbol{w}}_u^{\left( {k + 1} \right)} = \mathop {\arg \min }\limits_{{{\boldsymbol{w}}_u}} \left( {\ell \left( {{{\boldsymbol{w}}_u},{{\boldsymbol{x}}^{\left( k \right)}},{y^{\left( k \right)}}} \right)} \right), \
k = 1,2,\cdots, K,
\end{equation}
where $\boldsymbol{w}_u^{(k+1)}$ denotes the vector of the learned user preference model parameters of the $k$th iteration for the user $u$.
By using {the online gradient descent (OGD) method}  \cite{OGD}, the solution of the above optimization problem can be obtained through the iteratively updated model parameters as follows
\begin{align}\label{eq-18}
{\boldsymbol{w}}_u^{\left( {k + 1} \right)} = {\boldsymbol{w}}_u^{\left( k \right)} - {\eta ^{\left( k \right)}}{{\boldsymbol{g}}^{\left( k \right)}}, \ k=1,2,\cdots, K,
\end{align}
where ${{\boldsymbol{g}}^{\left( k \right)}}=\nabla _{\boldsymbol{w}_u} \ell \left( {{{\boldsymbol{w}}_u},{{\boldsymbol{x}}^{\left( k \right)}},{y^{\left( k \right)}}} \right)=\left[ {{p_{\boldsymbol{w}_u}}\left( {\left. {{y^{\left( k \right)}}} \right|{{\boldsymbol{x}}^{\left( k \right)}}} \right) - {y^{\left( k \right)}}} \right]{{\boldsymbol{x}}^{\left( k \right)}}$ denotes the gradient vector of the logistic loss of the $k$th sample with respect to $\boldsymbol{w}_u$,
and ${{\eta ^{\left( k \right)}}}$  a non-increasing learning-rate schedule with  $\sigma ^{(k')}$ satisfying $\sum\nolimits_{k' = 1}^k {{\sigma ^{\left( k' \right)}}}  = {1}/{{{\eta ^{\left( k \right)}}}}$.
Then,
{we have the following theorem.
\begin{theorem}\label{th0}
The   solution of the optimization problem in \eqref{eq-10} can also be obtained by solving the following equivalent optimization problem
\begin{multline}\label{eqeq-10}
 {\boldsymbol{w}}_u^{\left( {k + 1} \right)} =
 \mathop {\arg \min }\limits_{{{\boldsymbol{w}}_u}} \left( {{\left( {{{\boldsymbol{g}}^{\left( {1:k} \right)}}} \right)}^{\rm T}}  {{\boldsymbol{w}}_u} \right.\\
 \left. + \frac{1}{2}\sum\limits_{k' = 1}^k {{\sigma ^{\left(k' \right)}}\left\| {{{\boldsymbol{w}}_u} - {\boldsymbol{w}}_u^{\left( k' \right)}} \right\|_2^2}  \right), \
 k = 1,2, \cdots ,K,
\end{multline}
where ${{\boldsymbol{g}}^{\left( {1:k} \right)}} = \sum\nolimits_{k' = 1}^k {{{\boldsymbol{g}}^{\left( k' \right)}}} $ denotes the sum of the gradient  vector of the logistic loss of the previous $k$ samples.
\end{theorem}}
\begin{proof}
Please see Appendix A.
\end{proof}

Due to the sparse and unbalanced data and high-dimensional feature vector, there may exist over-fitting and high computational complexity problems \cite{Learning}.
Inspired by  {the follow the (proximally) regularized leader (FTRL-Proximal) method} \cite{FTRL}, which is an online optimization method based on the OGD method,
the $L1$-  and $L2$-regularization terms are introduced simultaneously in the optimization problem in \eqref{eqeq-10}
in order to avoid the above mentioned potential problems whereas obtain the optimal model parameters.
The introduction of the $L1$-regularization term is beneficial for realizing feature selection and producing sparse model, while the introduction of the $L2$-regularization term is conductive to the smooth solution of  the corresponding optimization problem.
Correspondingly, the model parameters can be updated according to the previous samples by solving the following optimization problem
\begin{multline}\label{eq11}
{{\boldsymbol{w}}_u^{\left( {k + 1} \right)}}
{\rm{ = }}\mathop {{\rm{argmin}}}\limits_{{{\boldsymbol{w}}_u}} \left\{
{( {{{\boldsymbol{g}}^{\left( {1:k} \right)}} - \sum\limits_{k' = 1}^k {{\sigma ^{\left( k' \right)}}{\boldsymbol{w}}_u^{\left( k' \right)}} } )^T  {{\boldsymbol{w}}_u}}
 \right. \\
 { + \frac{1}{2}( {{\lambda _2} + \sum\limits_{k' = 1}^k {{\sigma ^{\left( k' \right)}}} } )\| {{{\boldsymbol{w}}_u}} \|_2^2 } + {\lambda _1}{{\| {{{\boldsymbol{w}}_u}} \|}_1}    \\
\left. + \frac{1}{2}\sum\limits_{k' = 1}^k {{\sigma ^{\left( k'\right)}}\| {{\boldsymbol{w}}_u^{\left( k' \right)}} \|_2^2}\right\}, \  k = 1,2,\cdots, K,
\end{multline}
where ${\lambda _1}$ and $ {\lambda _2}$  denote the  regularization parameters with positive values, ${\|  \cdot  \|_1}$ and $\|  \cdot  \|_2^2$ the $L1$-norm and $L2$-norm, respectively.

It can be readily seen from \eqref{eq11} that the last item in the right hand side of \eqref{eq11}, i.e., ${1}/{2}\sum\nolimits_{k' = 1}^k {{\sigma ^{\left( k' \right)}}\|
{{\boldsymbol{w}}_u^{\left( k' \right)}} \|_2^2} $, is irrespective with ${\boldsymbol{w}_u}$.
Let
$
{{\boldsymbol{z}}^{\left( k \right)}} = {{\boldsymbol{g}}^{\left( {1:k} \right)}} - \sum\nolimits_{k' = 1}^k {{\sigma ^{\left( k' \right)}}{\boldsymbol{w}}_u^{\left( k' \right)}}.
$
Then, an iterative relationship between ${{\boldsymbol{z}}^{\left( k \right)}}$ and ${{\boldsymbol{z}}^{\left( k-1 \right)}}$ can be established as follows
\begin{equation}
{{\boldsymbol{z}}^{\left( k \right)}} = {{\boldsymbol{z}}^{\left( {k - 1} \right)}} + {{\boldsymbol{g}}^{\left( k \right)}} - \left( {\frac{1}{{{\eta ^{\left( k \right)}}}} - \frac{1}{{{\eta ^{\left( {k - 1} \right)}}}}} \right) {\boldsymbol{w}}_u^{\left( k \right)},
\end{equation}
which implies that we only need to store ${{\boldsymbol{z}}^{\left( k-1 \right)}}$ after using the last sample for learning.
Correspondingly, the  optimization problem in \eqref{eq11} can be further expressed  as follows
\begin{multline}\label{eq12}
{\boldsymbol{w}}_u^{\left( {k + 1} \right)} = \mathop {{\rm{argmin}}}\limits_{{{\boldsymbol{w}}_u}} \left\{
{({{\boldsymbol{z}}^{\left( k \right)}})^T  {{\boldsymbol{w}}_u} + {\lambda _1}{{\| {{{\boldsymbol{w}}_u}} \|}_1}} \right.\\
 \left.{ + \frac{1}{2}( {{\lambda _2} + \sum\limits_{k' = 1}^k {{\sigma ^{\left(k' \right)}}} } )\| {{{\boldsymbol{w}}_u}} \|_2^2}
 \right\}, \ k = 1,2,\cdots, K.
\end{multline}

We know that there exists a difference for the change rate of the weight of each feature dimension for the requested content, and the gradient value with respect to each feature dimension can reflect this change rate.
Therefore,  different learning rates are preferred for different feature dimensions.
Define
\begin{align}
\boldsymbol{g}^{(k)} &= [ {{g_{1}^{(k)}},{g_{2}^{(k)}},\cdots,{g_{n}^{(k)}},\cdots, {g_{N}^{(k)}}} ]^T,\\
\boldsymbol{z}^{(k)} &= [ {{z_{1}^{(k)}},{z_{2}^{(k)}},\cdots,{z_{n}^{(k)}},\cdots, {z_{N}^{(k)}}} ]^T, \\
{{\boldsymbol{w}}_u} &= [ {{w_{u,1}},{w_{u,2}},\cdots,{w_{u,n}},\cdots, {w_{u,N}}} ]^T,\\
{{\boldsymbol{w}}_u^{(k+1)}} &  = [ {{w_{u,1}^{(k+1)}},{w_{u,2}^{(k+1)}},\cdots,{w_{u,n}^{(k+1)}},\cdots, {w_{u,N}^{(k+1)}}} ]^T.
\end{align}
Let
$
\sum\nolimits_{k' = 1}^k {\sigma _n^{\left( k' \right)}}  = \frac{1}{{\eta _n^{\left( k \right)}}},
$
where
$\eta _n^{\left( k \right)} = {\alpha } /[{\beta  + \sqrt {\sum\nolimits_{k' = 1}^k {{{( {g_n^{( k' )}} )}^2}} } } ]  $
denotes the learning-rate schedule of the $n$th feature dimension with $\alpha$  and $\beta$ being the adjusting parameters which are  chosen to yield  good learning performance.
Then, the  optimization problem in \eqref{eq12} can be decoupled into the following $N$ independent scalar minimization problems
\begin{multline}\label{eq13}
{w_{u,n}^{(k+1)}} = \mathop {{\rm{argmin}}}\limits_{{w_{u,n}}  } \left\{ z_n^{\left( k \right)}  {w_{u,n}} + {\lambda _1}\left| {{w_{u,n}}} \right|  \right.\\
\left.+ \frac{1}{2}( {{\lambda _2} + \sum\limits_{k' = 1}^k {\sigma _n^{\left( k' \right)}} } )w_{u,n}^2  \right\}, \ n = 1, 2, \cdots, N.
\end{multline}

It can be easily verified that the optimization problem in \eqref{eq13} is an unconstrained non-smooth one, where the second item ${\lambda _1}\left| {{w_{u,n}}} \right|$ is non-differentiable at ${w_{u,n}} = 0$. Let $\eta  = \partial \left| {w_{u,n}} \right| \left| _{w_{u,n}= w_{u,n}^{(k+1)}} \right.$ denote the partial differential of $\left| {w_{u,n}} \right|$ at ${w_{u,n}^{\left( {k + 1} \right)}}$. Then, we have
\begin{align}\label{eq14}
\left\{ {\begin{array}{*{20}{l}}
{ { - 1 < \eta  < 1}, }&{{\rm{if }}\ {w_{u,n}^{(k+1)}} = 0,}\\
{ \eta = -1, }&{{\rm{else\ if }}\ {w_{u,n}^{(k+1)}} < 0,}\\
{ { \eta =  1,} }&\rm{  otherwise}.
\end{array}} \right.
\end{align}
Correspondingly, the optimal solution ${w_{u,n}^{(k+1)}}$ should  satisfy the following relationship
\begin{multline}\label{eq15}
{z_n^{\left( k \right)} + {\lambda _1}\eta  + ( {{\lambda _2} + \sum\limits_{k' = 1}^k {\sigma _n^{\left( k' \right)}} } )w_{u,n}^{\left( {k + 1} \right)} = 0}, \\
 n = 1, 2, \cdots, N.
\end{multline}
We have known previously that ${\lambda _1} > 0$.
Correspondingly, by classifying $z_n^{(k)}$ into three cases, i.e., $| {z_n^{\left( k \right)}} | < {\lambda _1}$, $z_n^{\left( k \right)} > {\lambda _1}$ and $z_n^{\left( k \right)} <  - {\lambda _1}$,
the closed-form solution of the optimization problem in \eqref{eq13} can be obtained from \eqref{eq15} as follows
\begin{multline}\label{eq16}
w_{u,n}^{\left( {k + 1} \right)} = \left\{ {\begin{array}{*{20}{l}}
0,&{{\rm{if}} \ | {z_n^{\left( k \right)}} | < {\lambda _1}},\\
{  \frac{{ {{\lambda _1}{\rm{sgn}}( {z_n^{\left( k \right)}} )- z_n^{\left( k \right)} } }}{{ {{\lambda _2} + \sum\nolimits_{k' = 1}^k {{\sigma _n ^{\left( k' \right)}}} } }}},&{{\rm{otherwise}}},
\end{array}} \right. \\
n = 1, 2, \cdots, N.
\end{multline}

\begin{algorithm}[!t]
\small
	\renewcommand{\algorithmicrequire}{\textbf{Input:}}
	\renewcommand{\algorithmicensure}{\textbf{Output:}}
	\caption{The offline user preference learning algorithm}
	\label{alg:2}
	\begin{algorithmic}[1]
		\Procedure {Learn}{$\left\{ {\left( {{{\boldsymbol{x}}^{\left( k \right)}},{y^{\left( k \right)}}} \right)} \right\}_{k = 1}^K$}
        \State Initialize: $\alpha$, $\beta$, ${\lambda _1}$, ${\lambda _2}$, ${\boldsymbol{w}}_u^{\left( 1 \right)} $, ${{\boldsymbol{z}}^{\left( 0 \right)}}= {\boldsymbol{q}}^{\left( 0 \right)} = \boldsymbol 0 \in {\mathbb{R}^N}$;
		\For{$k = 1,2,3,\cdots,K$, }
        \State ${{\boldsymbol{g}}^{\left( k \right)}} = {\left. {\nabla {\ell _{{{\boldsymbol{w}}_u}}}\left( {{{\boldsymbol{w}}_u},{{\boldsymbol{x}}^{\left( k \right)}},{y^{\left( k \right)}}} \right)} \right|_{{{\boldsymbol{w}}_u} = {\boldsymbol{w}}_u^{\left( k \right)}}}$;
        \For{$n = 1,2,3,\cdots,N$, }
		\State $\sigma _n^{\left( k \right)} = \frac{1}{\alpha }( {\sqrt {q_n^{\left( {k - 1} \right)} + {{( {g_n^{( k )}} )}^2}}  - \sqrt {q_n^{\left( {k - 1} \right)}} } )$;
        $z_n^{\left( k \right)} = z_n^{\left( {k - 1} \right)} + g_n^{\left( k \right)} - \sigma _n^{\left( k \right)}w_n^{\left( k \right)}$;
         $q_n^{\left( k \right)} = q_n^{\left( {k - 1} \right)} + { ({g_n^{\left( k \right)}})^2 }$;
		\State Caculate $w_{u,n}^{\left( {k + 1} \right)}$  according to \eqref{eq16} by setting
         $\sum\nolimits_{r = 1}^k {{\sigma _n ^{\left( r \right)}}}$  to  $({{\beta {\rm{ + }}\sqrt {q_n^{\left( k \right)}} }})/{\alpha }$.
        \EndFor
        \EndFor
        \State \textbf{return} ${\boldsymbol{w}}_u^{\left( {K + 1} \right)}$
        \EndProcedure
	\end{algorithmic}
\end{algorithm}

The entire user preference learning algorithm with the property of self-starting is described in Algorithm 3.
Note that our proposed  algorithm  only needs to store the last ${\boldsymbol{w}}_u^{\left( K \right)}$ and ${{\boldsymbol{z}}^{\left( K \right)}}$ which will be the initialized parameters for the next user preference model updating, and the previously recorded data can be cleared which is helpful to save storage space. Furthermore, the time complexity of our proposed algorithm for one user preference model updating is ${\mathcal O}\left( {KN} \right)$, which is not an issue due to its offline property.

We remark here that our proposed  policy can asymptotically  approach the optimal solution of the optimization problem
in \eqref{eq_1}
as the user requests increase, whose proof will be presented in Section V.

\section{The Proposed Learning based Edge Caching Architectures}

In this section, we propose two learning based edge caching architectures (as illustrated in Fig. \ref{fig2} and Fig. \ref{fig2-new}) which can implement the functionality  of our previously proposed edge caching policy.
In our proposed first architecture, by considering that not all UEs are equipped with artificial intelligence (AI) chipsets and support offline learning, both the online popularity prediction algorithm and the offline user preference learning algorithm are implemented inside the F-APs. In our proposed second architecture,  by considering future UEs equipped with AI chipsets supporting offline learning in smart wireless communications scenarios, the online popularity prediction algorithm is implemented inside the F-APs and the offline user preference learning algorithm is implemented inside the UEs. These two architectures will be presented in detail below.


\subsection{Learning Based Edge Caching Architecture (I)}


\begin{figure}[!t]
\centering
\includegraphics[width=0.48\textwidth]{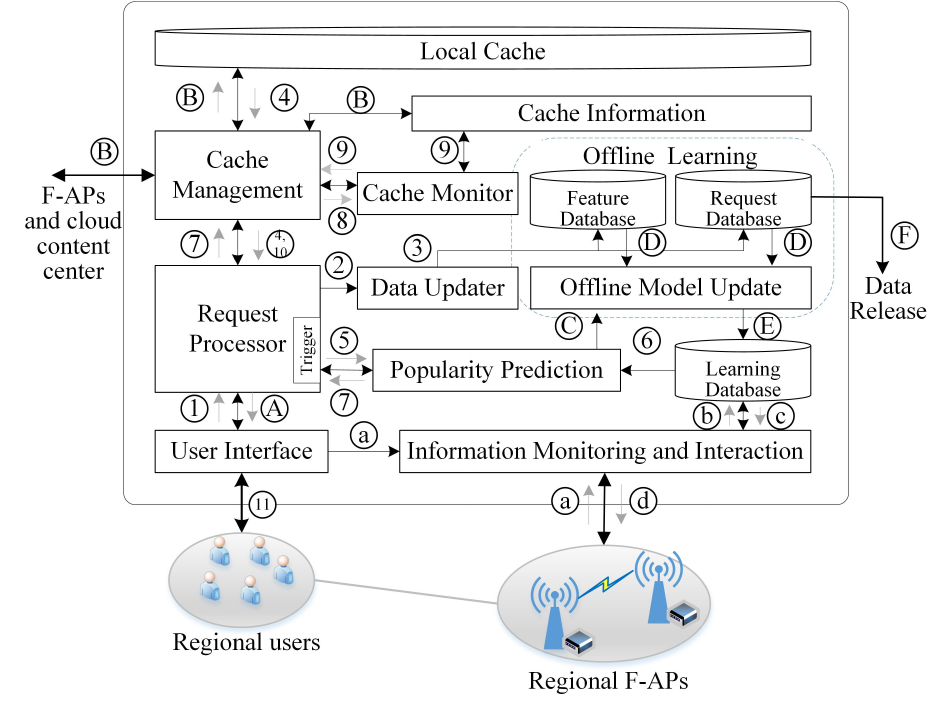}
\caption{Illustration of the first learning based edge caching architecture.}
\label{fig2}
\end{figure}
\begin{figure}[!t]
\centering
\includegraphics[width=0.38\textwidth]{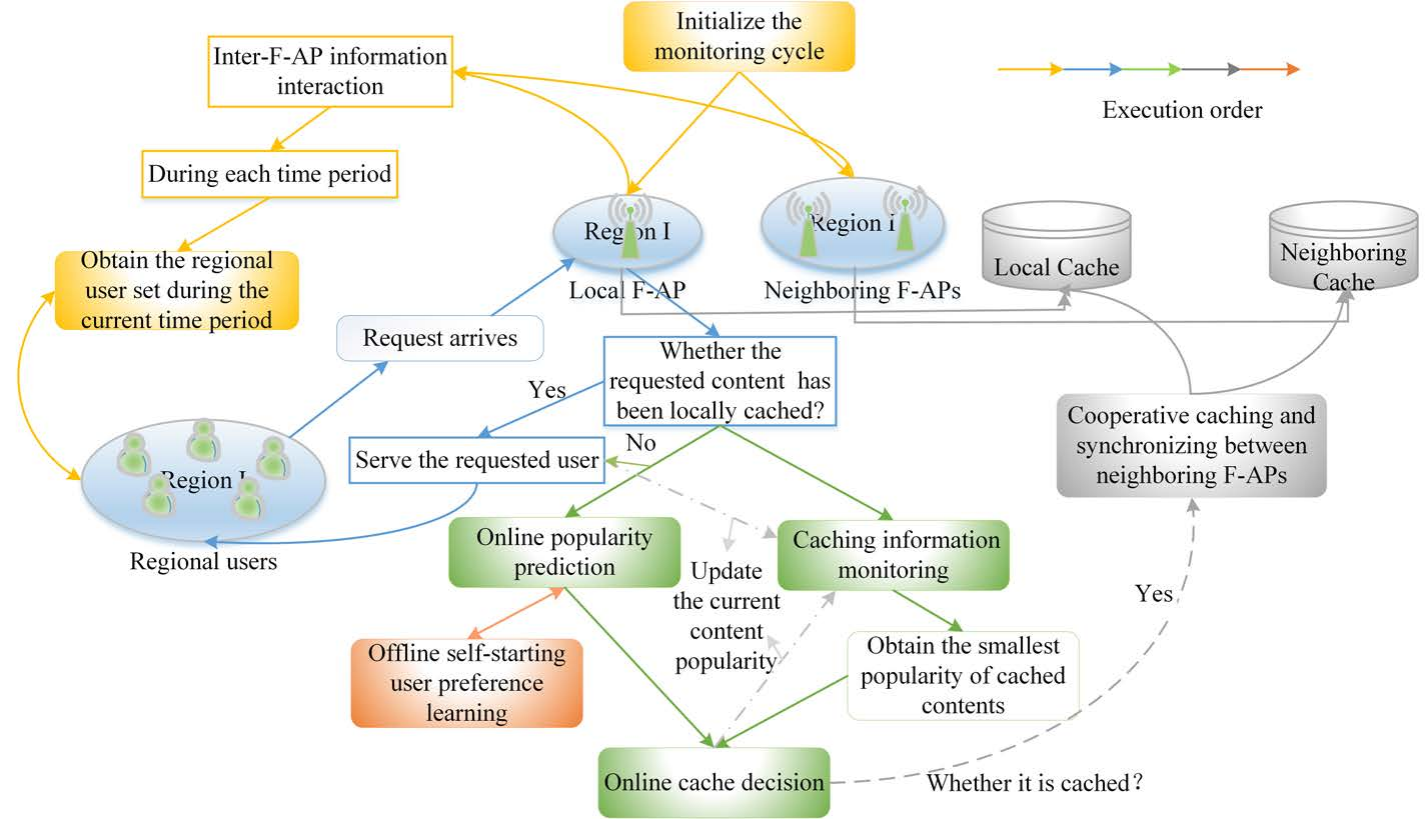}
\caption{\normalsize Illustration of the second learning based edge caching architecture.}
\label{fig2-new}
\end{figure}

\begin{figure*}[!t]
\centering 
\includegraphics[width=0.7\textwidth]{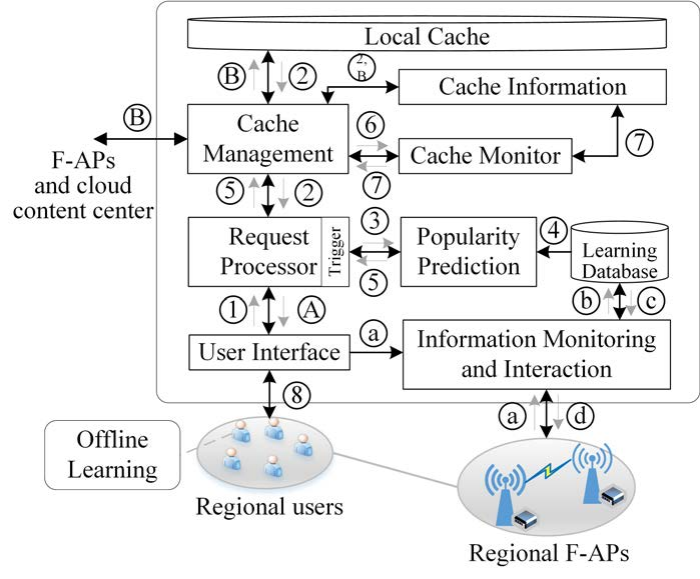}
\caption{Illustration of the learning based edge caching flowchart.}
\label{fig3}
\end{figure*}

For the proposed first architecture as illustrated in Fig. \ref{fig2}, its fundamental modules are as follows:  \emph{Local Cache}, \emph{Cache Management}, \emph{Request Processor}, and \emph{User Interface}, which have functions similar to  the traditional caching architectures \cite{CMAB, Trend}.
In order to  learn  user preference and predict content popularity, our proposed architecture also includes the following modules: \emph{Information Monitoring
and Interaction}, \emph{Popularity Prediction}, \emph{Offline Learning}, \emph{Data Updater}, \emph{Cache Information}, and \emph{Cache Monitor}. Their functions are described as follows.

\begin{itemize}
\item The {Information Monitoring and Interaction} module is mainly responsible for realizing
regular information monitoring and interaction between regional F-APs.
On the one hand, this module periodically collects the current
user set of the serving F-AP and the current user information (including the regional user set and user preference model) of the other regional F-APs, and stores them into the {Learning Database}.
On the other hand, this module  periodically sends the current  user information of the serving F-AP to the other regional ones, and finally realizes the monitoring and sharing of the current  user information among the regional F-APs.
\item The Popularity Prediction module is mainly responsible for predicting the  popularity of the current requested content based on the Learning Database.
    Note that the Offline Learning module will be initiated if the average   prediction error is larger than {the predefined threshold}.
\item The {Offline Learning module} is mainly responsible for learning the current user preference model parameters based on the collected information from the {Feature Database} and {Request Database}.
\item The {Data Updater} module is mainly responsible for updating the content feature data,
the requested content, and the requested time to the Feature Database and Request Database, and realizing the collection and update of the requested information.
\item The {Cache Information} module is mainly responsible for storing and updating the current content popularity information, the initial cache time, and the cache content ID.
\item The {Cache Monitor} module is mainly responsible for monitoring the cache information to capture the contents which need to be removed from the  local cache.
\end{itemize}

The flowchart of our proposed first learning base edge caching architecture consists of five phases as illustrated in Fig. \ref{fig2} and Fig. \ref{fig3}, and is presented below.

\begin{enumerate}
    \item[(i)] {Initializing and periodic information monitoring}

(a) The Information Monitoring and Interaction module periodically extracts the current  user information of the serving F-AP and regional ones from their User Interface modules.
(b) This module regularly updates the collected regional user information to the Learning Database.
(c) This module extracts the current  user information of the serving F-AP from the {Learning Database}.
(d) This module delivers the current  user information of the serving F-AP to the regional ones.
    \item[(ii)] {Direct local request responding}

(1)  The {User Interface} module delivers the user requesting information to the {Request Processor} module. 
(2)  {The {Request Processor} module} initiates a data updating procedure. 
(3)  The {Data Updater} module carries on numerical processing to the requested content feature and writes the processed feature data into the {Feature Database}, and updates the requested content information in the {Request Database}.
(4)  If the requested content is  stored locally, the {Cache Management} module delivers the stored content from the local cache to the {Request Processor} module.
(A)  The {Request Processor} module serves the user request.
    \item[(iii)] {Dynamic content caching and updating}

(5) If the requested content is not stored locally, the {Request Processor} module triggers the {Popularity Prediction} module to make online content popularity prediction of the requested content.
(6) The {Popularity Prediction} module extracts the user preference model parameters and the requested content features from the {Learning Database}, and then predicts the popularity of the requested content.
(7)  The {Popularity Prediction} module feeds back the predicted popularity of the requested content first to the {Request Processor} module, and then to the {Cache Management} module through the {Request Processor} module.
(8)  The {Cache Management} module triggers the {Cache Monitor} module to initiate monitoring of the cache content information.
(9)  The {Cache Monitor} module extracts the information of the content to be removed from the {Cache Information} module, and then feeds back the information to the {Cache Management} module.
(10)  The {Cache Management} module makes cache decision based on the feedback information, and notifies the Request Processor module to serve the user request.
(11)  The User Interface module broadcasts the local cached content information to the  regional users.
     \item[(iv)] {Cooperative caching and information synchronizing}

(B)  The {Cache Management} module executes the received cache decision of  the current F-AP, and notifies the other regional F-APs  to execute the same cache decision. If the requested content is  to be cached, the {Cache Management} module retrieves the content from the  neighboring F-APs in the other regions or the cloud content center, and then stores it locally by means of partition-based caching.
The {Cache Management} module updates the cache information of the current F-AP, and synchronizes the cache information of the other regional F-APs.
     \item[(v)] Offline self-starting user preference learning

(C)  The {Popularity Prediction} module initiates the {Offline Learning} module if  the average prediction error cumulated under a user preference model is larger than {the predefined threshold.}
(D)  The {Offline Model Update} module retrieves the historical requested data of the considered user from the {Feature Database} and {Request Database}, {generates a new training sample set,} and then updates the user preference model parameters.
(E)  The Offline Model Update module refreshes the updated user preference model parameters to the {Learning Database}.
(F)  {The Request Database releases the  historical requested data of the considered user.}
\end{enumerate}


\subsection{Learning Based Edge Caching Architecture (II)}


For the proposed second architecture as illustrated in Fig. \ref{fig2-new},
part of the functionality, i.e., offline user preference learning, is transferred from the F-APs to the  powerful and smart UEs. Therefore,
both the signaling overhead among regional F-APs and the computational burden undertaken by the F-APs can be greatly reduced.

The corresponding processing flow of our proposed edge caching policy is presented in brief below.
(a-d) The current F-AP is mainly responsible for monitoring the regional users in coordination with the regional F-APs and storing the corresponding user preference model parameters from the UEs into the Learning Database. (1-2, A) The current F-AP serves the requested user if it caches the requested content locally. (3-8) If the requested content is not stored in the local cache, the current F-AP predicts the content popularity, makes the corresponding caching decision, and broadcasts the caching information to the regional users. (B) The current F-AP notifies the other regional F-APs to execute the same cache decision and {then updates} the corresponding cache information.

We remark here that our proposed first architecture is more suitable for  wireless communications scenarios including intelligent F-APs and general UEs while our proposed second architecture is more suitable for  wireless communications scenarios including general F-APs and intelligent UEs. With the rapid development of AI and smart UEs, the second architecture will show more advantages and dominate in future wireless communications.
We also remark here that the signaling overhead can be further reduced by setting a cluster head for the clustered regional F-APs and the corresponding edge caching architecture is omitted here due to space limitation.

\section{Performance Analysis}
    In this section, the performance of our proposed edge caching policy will
be analyzed.
Firstly, we derive the upper bound of the popularity prediction error of
our proposed online content popularity prediction algorithm.
Secondly, we  derive the lower bound of the cache hit rate of our proposed edge caching policy.
Finally, we derive the regret bound of the overall cache hit rate of our proposed edge caching policy.

\vspace*{-10pt}
\subsection{The Upper Bound of the Popularity Prediction Error}

Let $\mathcal{E} $ denote the convex set of the optimization problem in \eqref{eq11} with $\mathcal{E} \in {\mathbb{R}^N}$, and define
\begin{equation}
{W_u} = \mathop {{\max}}\limits_{\boldsymbol{w}_u,\boldsymbol{w}_{u'} \in \mathcal{E}} \left\| {{{\boldsymbol{w}}_u} - {{\boldsymbol{w}}_{u'}}} \right\|.
\end{equation}
Let ${{L_d}}\left( {{{\boldsymbol{w}}_u}} \right)$ denote a sequence of convex loss functions,
and define
\begin{equation}
{G_u} = \mathop {\max }\limits_{{w_u} \in \mathcal{E},1\le d \le D_t, t = 1, 2, \cdots, T } 
\left\| {\nabla {L_d}\left(
{{{\boldsymbol{w}}_u}} \right)} \right\|.
\end{equation}
Then, according to Corollary 1 in \cite{Regret},
the following relationship can be readily established for the optimization problem in \eqref{eq11} in the finite time horizon $T$
\begin{align}\label{eq17}
\sum\limits_{t = 1}^T {\sum\limits_{d = 1}^{{D_t}} {{L_d}\left( {{{\boldsymbol{w}}_u}} \right) - \sum\limits_{t = 1}^T {\sum\limits_{d = 1}^{{D_t}} {{L_d}\left( {{\boldsymbol{w}}_u^*} \right)} } } }  \le {W_u}{G_u}\sqrt {2{D}},
\end{align}
where ${{\boldsymbol{w}}_u^*}$ denotes the vector of the optimal user preference model parameters.
Let ${\tau _u}$ denote a sufficiently small constant with a positive value, which can meet  $\sum\limits_{t = 1}^T {\sum\limits_{d = 1}^{{D_t}} {{L_d}\left( {{\boldsymbol{w}}_u^*} \right)} }  \le {\tau _u}$.
Define
\begin{align}
{W_{\max }} &= \mathop {{\max}}\limits_{u \in {U_t}, t=1,2,\cdots, T} {{W_u}},\\
{G_{\max }} &= \mathop {{\max}}\limits_{u \in {U_t}, t=1,2,\cdots, T} {{G_u}},\\
{\tau _{\max }} &= \mathop {{\max}}\limits_{u \in {U_t}, t=1,2,\cdots, T} {{\tau_u}}.
\end{align}
Then, we have the following theorem.
\begin{theorem}\label{th1}
The expected popularity prediction error for the overall $D$ requests in the finite time horizon $T$, i.e., $\mathbb{E}\sum\limits_{t = 1}^T {\sum\limits_{d = 1}^{{D_t}} {\left| {{{\hat P}_{t,d}} - {P_{t,d}}} \right|} }$, can be upper bounded by $\frac{{U_{\max }}}{{{U_{\min }}}}\left( {W_{\max }}{G_{\max }}\sqrt {2D}  + {\tau _{\max }}\right)$.
\end{theorem}
\begin{proof}
Please see Appendix B.
\end{proof}

It is clear that an upper bound exists for the
cumulative prediction error of the  content requested probability of one single user  and  that of the  regional content popularity within a limited time periods.
Specifically, by using the user preference model that is obtained through self-learning, the upper bound of the cumulative prediction error of the  content requested probability of one single user  has a  sub-linear relationship with the total number of requests $D$ in the finite time horizon $T$.
Similar relationship can be found for the
cumulative prediction error of the  content popularity in the considered region, which means that $\frac{1}{D} \mathbb{E}\sum\limits_{t = 1}^T {\sum\limits_{d = 1}^{{D_t}} {\left| {{{\hat P}_{t,d}} - {P_{t,d}}} \right|} }\to 0 $ as $D \to \infty $.

The above analytical results  imply that the learned user preference model will asymptotically approach the  real user preference model through sufficient learning with the  collection of more user requesting information and the increased requests.
Correspondingly, the proposed online content popularity prediction algorithm can make the  content popularity prediction   more accurate.
On the other hand, the analytical results also reveal that the performance of our proposed policy can be improved with the increased  requests.
After a certain number of content requests, the prediction precision of the proposed policy can achieve the ideal value.

\subsection{The Lower Bound of the Cache Hit Rate}

In this subsection, we  first show the cache hit rate of the optimal edge caching policy which knows the real popularities of all the contents  and caches the most popular contents during each time period,
and then derive the lower bound of the cache hit rate of the proposed edge caching policy and reveal  their relationship. 

\subsubsection{The cache hit rate of the optimal edge caching policy}

In the ideal case, the optimal cache hit rate can be achieved by caching the most popular contents of the current time period based on
the known content popularity.
Let ${\Phi ^*}$ denote the optimal edge caching policy. In practice, we note that the cache hit rate depends not only on the edge caching policy ${\Phi }$ but also on {the degree of concentration of the same content requests.} Generally speaking, a more concentrated content request process implies a higher potential cache hit rate. For the ease of analysis, we assume that the requests of the same content  {are} concentrated in one time period.
Let $P_{t,d'}^{'}$ denote the content popularity of the ${d'}$th most popular content when it is requested firstly during the $t$th time period, {and $F_t$ the number of requested different contents} during the $t$th time period. Then, by using the relationship in \eqref{eq5}, the optimal cache hit rate during the  $t$th time period can be calculated as follows
\begin{align}\label{eq20}
\mathcal{H}_t \left( \Phi ^* \right) = \frac{{\sum\nolimits_{d' = 1}^{M\varphi } {U_t}{P_{t,d'}^{'}}} }{{\sum\nolimits_{d'= 1}^{F_t} {U_t}{P_{t,d'}^{'}}} }= \frac{{\sum\nolimits_{d' = 1}^{M\varphi } {P_{t,d'}^{'}}} }{{\sum\nolimits_{d' = 1}^{F_t} {P_{t,d'}^{'}}} }.
\end{align}

Note that $\mathcal{H}_t \left( \Phi ^* \right)$ has no relation with the number of requested users during the current time period, and may vary with different time periods.
Besides, we make no assumption about the popularity distribution, and $\mathcal{H}_t \left( \Phi ^* \right)$ is just the cache hit rate that is calculated based on the real content popularity during the current time period.

\subsubsection{The lower bound of the cache hit rate of the proposed edge caching policy}
In practice, there always exist popularity prediction errors.
Let $\hat P_{t,d'}^{'}$ denote the predicted   popularity with respect to $ P_{t,d'}^{'}$.
Let
\begin{equation}
\Delta {P_t} = \mathop {\max} \limits_{d'  = 1,2, \cdots ,F_t} {| {\hat P_{t,d'}^{'} - P_{t,d'}^{'} } |}.
\end{equation}
Then, we have the following theorem.

\begin{theorem}\label{th2}
During the $t$th time period, the achievable cache hit rate $\mathcal{H}_t \left( \Phi  \right)$ can be  lower bounded by $\mathcal{H}_t \left( \Phi ^* \right)  - {M\varphi} (\Delta {P_t} U_t + 2 )/ {D_t}$.
\end{theorem}
\begin{proof}
Please see Appendix C.
\end{proof}

{It is clear that a lower bound exists for the cache hit rate of the proposed edge caching policy during each time period. The analytical result from Theorem \ref{th2}  gives us the insight that there exists a certain performance gap, i.e., $ {{M\varphi }}\left( {\Delta {P_t}{U_t} + 2} \right)/ {{{D_t}}}$, between the cache hit rate of our proposed edge caching policy and that of the optimal edge caching policy.
The first component of the performance gap, i.e., $ {{M\varphi }}\Delta {P_t}{U_t}/ {{{D_t}}}$, is mainly caused by the
popularity prediction error $\Delta {P_t}$ and principally determined by the accuracy of the learned user preference model. The second component of the performance gap, {i.e., $ {{2M\varphi }}/{{{D_t}}}$,} is mainly caused by the operational mechanism of our proposed edge caching policy, which decides to cache a content only after but not before its first request (i.e., an initial cache miss will happen).

Although our proposed edge caching policy may result in an initial cache miss, it can actually avoid the extremely high computational load that may bring about to F-APs due to the need of continuously predicting all the content popularities otherwise. We point out here that this type of performance gap will gradually approach zero with the increased number of content requests, and hence the benefit outweighs the cost. Moreover, the analytical result from Theorem \ref{th2} also reveals that the overall performance gap will approach zero 
when the prediction error of the content popularity approaches zero, i.e., $\Delta {P_t} \to 0$, and the number of content requests during the $t$th time period is much larger than the overall cache size of all the F-APs in the considered region, i.e., ${D_t} \gg M\varphi $. Correspondingly, the cache hit rate of our proposed edge caching policy during the $t$th time period will approach that of the optimal edge caching policy if the above two conditions are satisfied.}

\subsection{The Regret Bound  of the Overall Cache Hit Rate}

In order to measure the performance loss of the proposed edge caching policy in comparison with the optimal one,
we will analyze and bound the regret  of the overall cache hit rate.

Let $\mathcal H \left( {{\Phi ^*}} \right)$ denote the overall cache hit rate of the optimal edge caching policy in the finite time horizon $T$. Then, from \eqref{eq_0} and \eqref{eq20}, it can be calculated as follows
\begin{equation}\label{eq21}
\mathcal H\left( {{\Phi ^*}} \right) = \frac{{\sum\nolimits_{t = 1}^T {\sum\nolimits_{d' = 1}^{M\varphi } {U_t} {P_{t, d'}^{{'}}} } }}{{\sum\nolimits_{t = 1}^T {\sum\nolimits_{d' = 1}^{F_t} {U_t} {P_{t,d'}^{{'}}} } }} = \frac{{\sum\nolimits_{t = 1}^T {\mathcal H_t \left( {{\Phi ^*}} \right){D_t}} }}{{\sum\nolimits_{t = 1}^T {{D_t}} }}.
\end{equation}
Let $\mathcal H \left( \Phi  \right)$ denote the overall cache hit rate of the proposed edge caching policy in the finite time horizon $T$. Then, we have
\begin{equation}\label{eq-new21}
\mathcal H\left( {{\Phi }} \right) = \frac{{\sum\nolimits_{t = 1}^T {\mathcal H_t \left( {{\Phi }} \right){D_t}} }}{{\sum\nolimits_{t = 1}^T {{D_t}} }}.
\end{equation}
Then, the regret of the overall cache hit rate of the proposed edge caching policy for the total $D$ requests in the finite time horizon $T$ can be defined as follows
\begin{equation}\label{eq39}
R\left( D \right) = \mathbb{E}\left[ {\mathcal H \left( {{\Phi ^*}} \right) - \mathcal H \left( \Phi  \right)} \right].
\end{equation}
Utilizing the analytical results from Theorem \ref{th1} and Theorem \ref{th2}, we have the following theorem.
\begin{theorem}\label{th3}
The regret of   the overall cache hit rate for the total $D$ requests in the finite time horizon $T$, i.e., $R\left( D \right)$, can be upper bounded by $ {{U_{\rm{max}}M\varphi }} \left[ {\frac{U_{\max }}{U_{\min } }\left({W_{\max }}{G_{\max }}\sqrt {2D}  + {\tau _{\max }}\right)  + 2T / U_{\rm{max}}} \right] /{D}$.
\end{theorem}
\begin{proof}
Please see Appendix D.
\end{proof}

According to the above theorem, with the limited $ U_{\max }$, $M\varphi$ and $ T$,
the following relationship can be naturally obtained:
$\mathop {\lim }\limits_{D \to + \infty } R(D) = 0,$
which shows that the performance loss of our proposed edge caching policy can be gradually reduced to zero as the number of requests is increased, i.e., our proposed edge caching policy has the capability to achieve the optimal performance asymptotically. The reason is that the learned
user preference model gradually approaches the real one with the increased request samples, which  makes the prediction errors even  smaller.

\section{Simulation Results}


To evaluate the performance of  the proposed edge caching policy, we take movie content\footnote{Other types of contents are also possible. However, due to the limitations of data acquisition, movie content is just taken here as an example.} as an  example and our main datasets are extracted from the MovieLens 200M Dataset \cite{MovieLens, Dataset}. 
From the MovieLens, we choose the requesting dataset of the selected 30 users\footnote{{It takes a great deal of simulation works to learn user preference models for a large number of users, and this is the reason why we only select 30 users in our simulations. In practice, regional F-APs indeed serve a much larger number of users. With the increase of the users served by the regional F-APs, the aggregation of the contents requested by the users will become much higher. Accordingly, popular contents will be more concentrated whereas non-popular contents will be more dispersed. Correspondingly, it will have a larger cache hit rate to cache the popular contents with more served users.}} who request the contents from January 01, 2010 to October 17, 2016.
The first part of the requesting dataset, whose requesting dates are from January 01, 2010 to December 31, 2015, 
is used for initializing the user preference,
while the second part of  the requesting dataset, whose requesting dates are from  January 01, 2016 to October 17, 2016, is used for evaluating the performance.
To simulate the content request process, we take the movie rating from a user as the request for this movie just as it is assumed in \cite{CMAB} and \cite{Trend}.
In our simulations, we randomly select 25 users from {the 30 users} as the fixed regional users while the remaining 5 users as the mobile users that randomly enter the region.
Besides, we set {the number of considered F-APs} $M$ to $3$, the finite time horizon $T$ to  $6984$ hours, the preset monitoring cycle to 1 hour\footnote{Note that we do not explore the impact of different monitoring cycle settings on our proposed caching policy. The reason is that this is a specific regional scenario which we simulate by randomly selecting users.
The impact of the monitoring cycle on the caching policy here has no reference significance to the actual scenario.} and the predefined threshold $\gamma$ to $0.2$, respectively.


\begin{figure}[!t]
\centering 
\includegraphics[height=4.5cm,width=7.5cm]{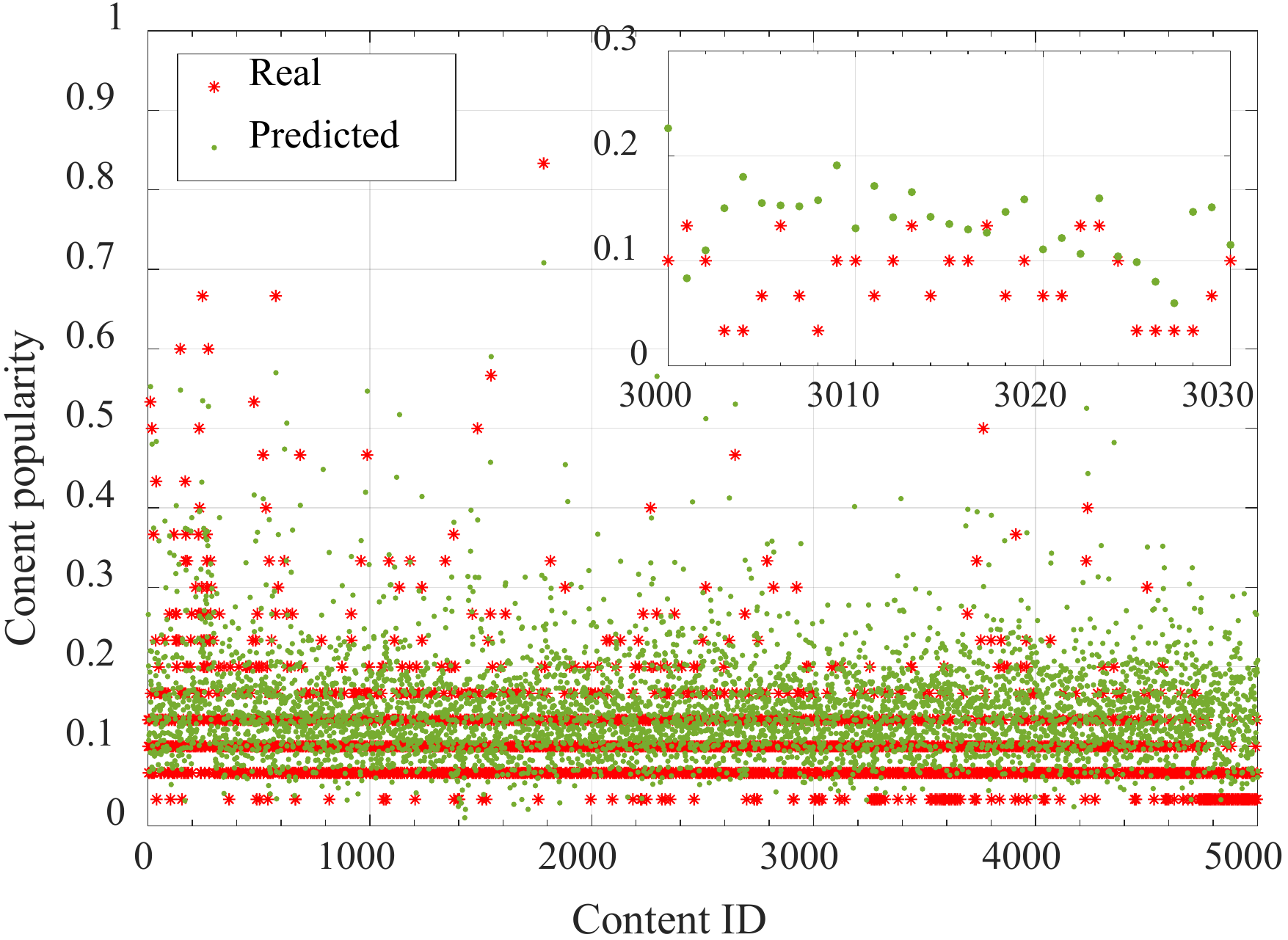}
\caption{Comparison of the predicted popularity and the real popularity in the finite time horizon $T$.}
\label{fig4}
\end{figure}


In Fig. \ref{fig4}, we show the  predicted popularity by using our proposed edge caching policy in comparison with the real popularity at the moment  when the  contents are  requested firstly by one of  the  regional users for the preceding 5000 contents in the finite time horizon $T$, where  the content ID is marked according to the first requesting  time of its representing content in chronological order.
It can be observed that the  error between the predicted popularity and the real popularity is very small. 
It can also be observed that the first requesting time of a regional popular content is random. Therefore, it is impractical for the existing edge caching policies to  cache the most popular contents directly without consideration of the content requesting time and temporal popularity dynamic.
It reveals the potential advantages of implementing a caching policy in conjunction with the content requesting time and real-time content popularity.

\begin{figure}[!t]
\centering
\includegraphics[height=4.5cm,width=7.5cm]{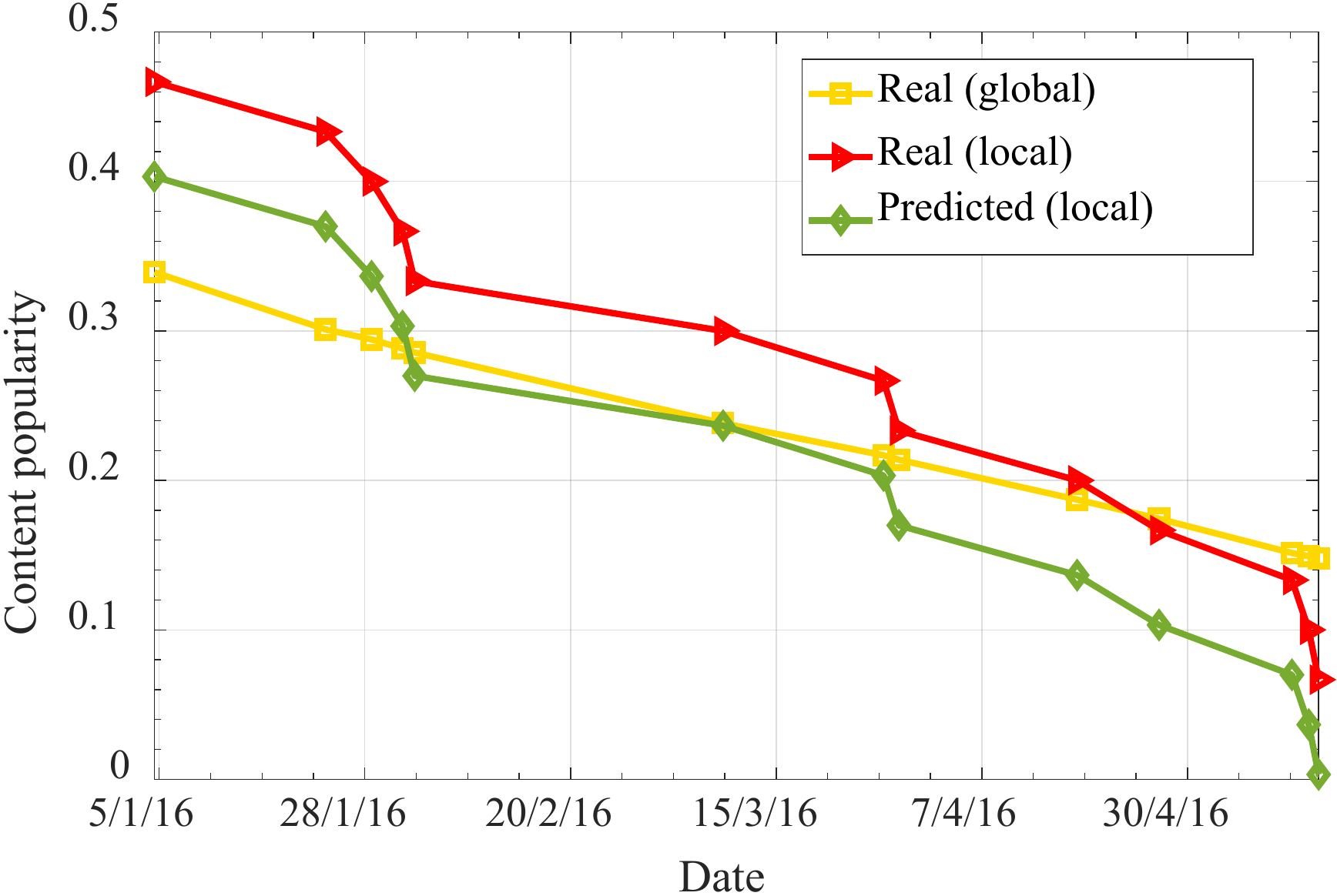}
\caption{Comparison of the predicted local popularity, the real local popularity, and the real global popularity for a certain requested content over time.}
\label{fig-new2}
\end{figure}

In Fig. \ref{fig-new2}, we  show the predicted local popularity, the real local popularity, and the real global popularity over time for a randomly selected content that is requested by the regional users.
It can be observed that both the global real  popularity and the local real popularity decrease with time whereas  the latter fluctuates slightly, which verifies that the global   popularity cannot precisely reflect the  temporal changes of the local   popularity. It can also be observed that the predicted local popularity changes with the real local   popularity in real time and the former approaches the latter, which reveals that our proposed policy can indeed track the real local popularity changes without delay.
This will guide the F-APs in a timely manner to clear the content that is no longer {popular}.

Without  consideration of  the requesting time difference  of the users and the duration difference of the continuous requests for the same content, we analyze the spatial changes of   content popularity.
In Fig. \ref{fig-new0}, we show the   predicted local popularity, the real local popularity for the considered region in comparison with the   real global popularity for all the regions, where  the content ID is marked according to  the real global popularity of its representing content in descending order.
It can be observed that most of the contents with the real local popularity larger than 0.2 have a content ID smaller that 2000, which reveals that most of the local popular contents originate from the global popular ones.
This also reveals that when a user preference model cannot be well learned due to sparse data, the global {popular} contents can be selected as the user's initial requesting contents.
We can observe that the contents with the content ID smaller than 500 generally have a larger real global popularity but a fluctuant real local popularity.
We can also observe that the real global popularity approximately follows a typical Zipf distribution whereas the real local popularity does not.
These observations reveal that the local popularity changes with the spatial popularity dynamic and does not necessarily follow the global popularity, and confirm the necessity of exploring the distribution of content popularity in a specific region.

\begin{figure}[!t]
\centering
\includegraphics[height=4.5cm,width=7.5cm]{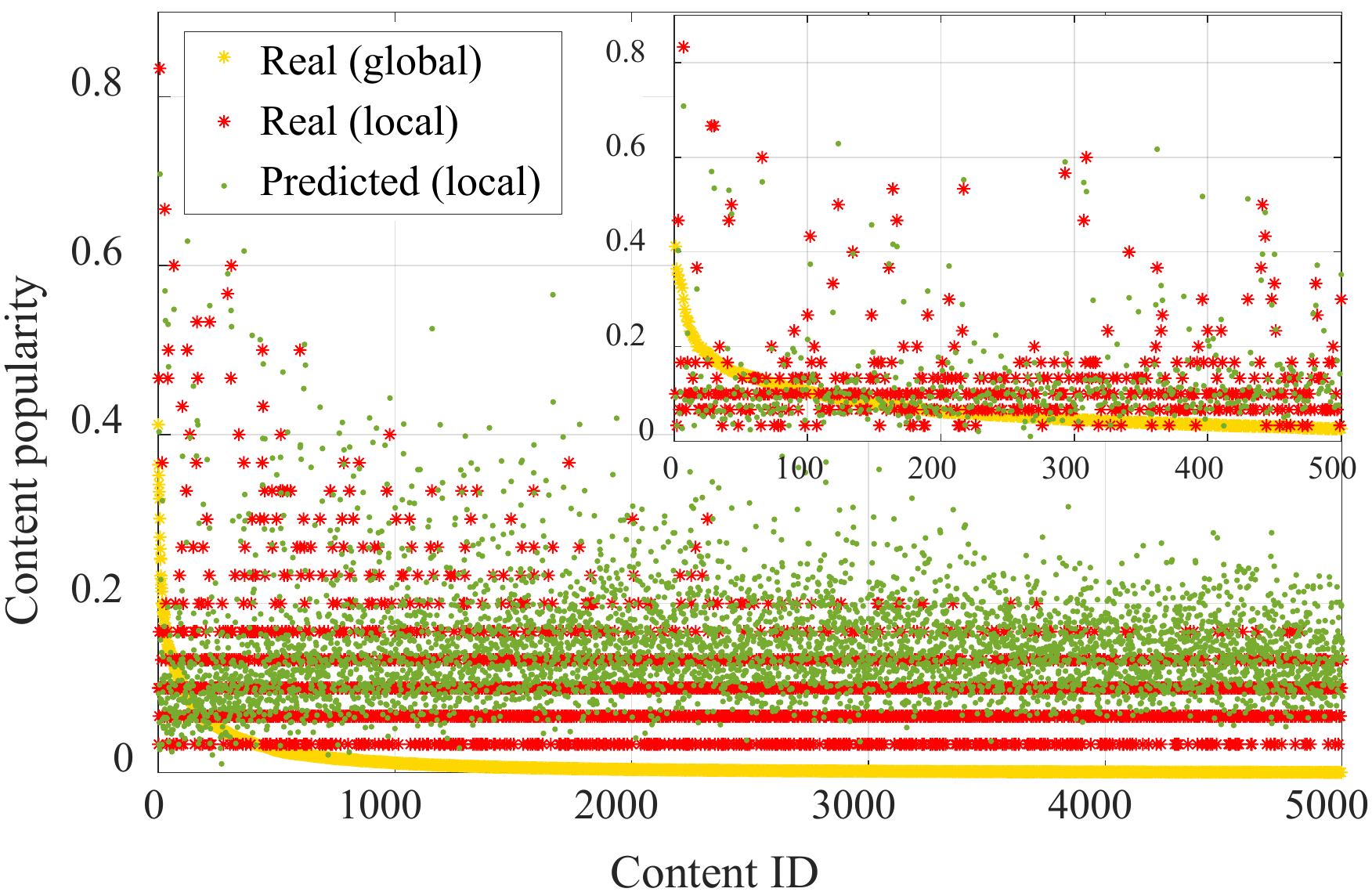}
\caption{\normalsize Comparison of the predicted local popularity, the real local popularity for the considered region, and the real global popularity for all the regions.}
\label{fig-new0}
\end{figure}


In Fig. \ref{hit}, we show the overall cache hit rate of our proposed policy with different $M \varphi$ in the finite time horizon $T$.
Also included in Fig. \ref{hit} are the overall cache hit rates of the four baseline policies, i.e., the FIFO  \cite{FIFO},  LRU  \cite{LRU}, LFU \cite{LFU} policies, and the optimal  policy with real content popularity.
The  total cache  size  $M \varphi$ increases from $1.5 \permil F =60$   to  $11.97\% F =4800$ contents with  $F  = 40110$.
It can be observed that the overall cache hit rates of all the considered policies are gradually increased with $M \varphi$.
It can also be observed that the overall cache hit rate of our proposed policy gradually approaches the optimal performance and is apparently superior to those of the other three baseline policies for all the considered cache sizes.
The reason is that the latter can not predict future content popularity.
Instead, our proposed  policy not only can predict the content popularity online,
but also can track its changes in real time.
Specifically,  it can be observed that our proposed policy only needs a  cache size of approximately 2400 contents to achieve the cache hit rate of 0.6 whereas the other three baseline policies need a cache size of approximately 4200 contents. 

%

\begin{figure}[!t]
\centering
\includegraphics[height=4.5cm,width=7.5cm]{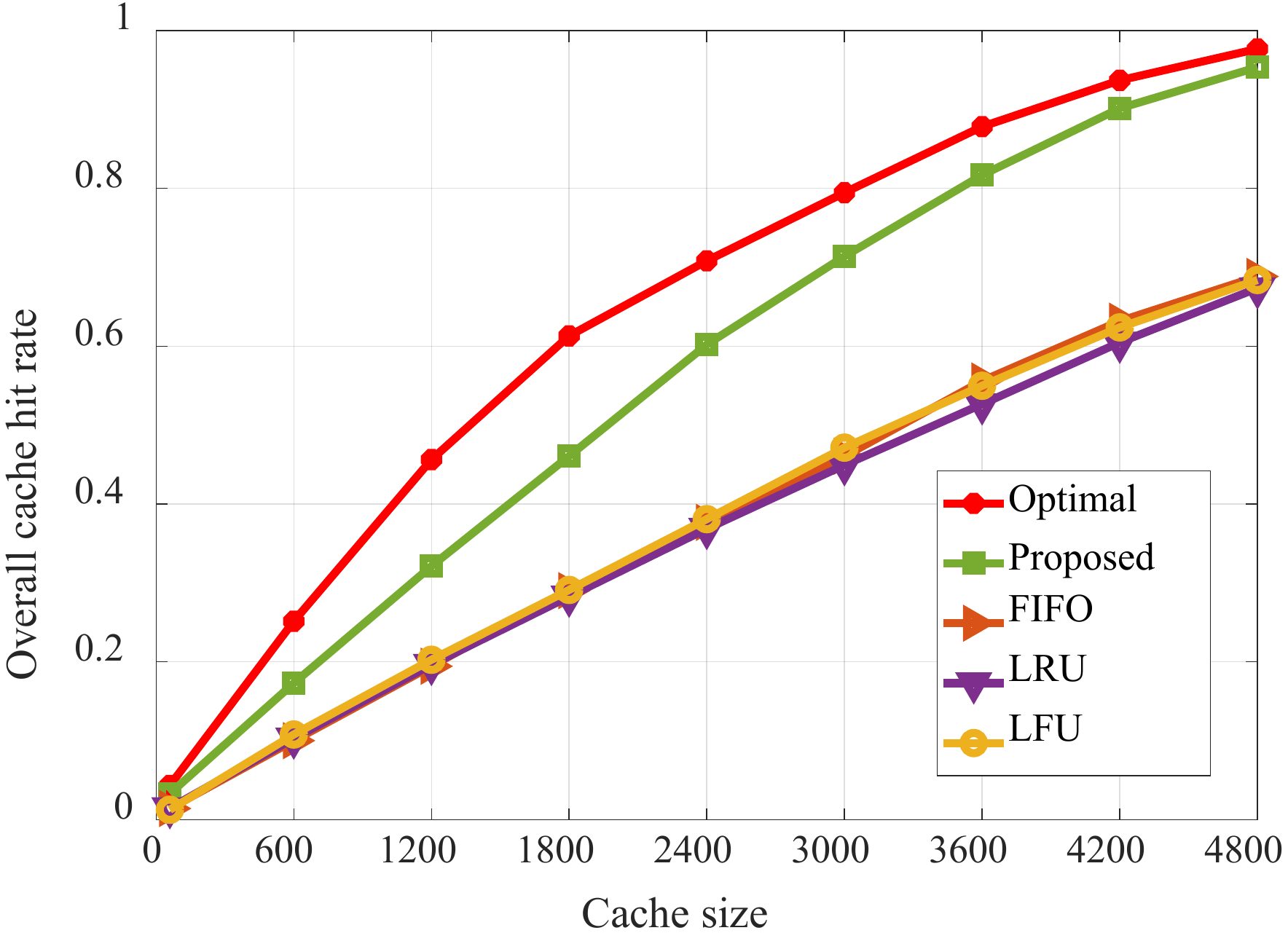}
\caption{The overall cache hit rate versus the total cache size ($M \varphi$) for the proposed policy and the baseline policies.}
\label{hit}
\end{figure}
\begin{figure}[!t]
\centering
\includegraphics[height=4.5cm,width=7.5cm]{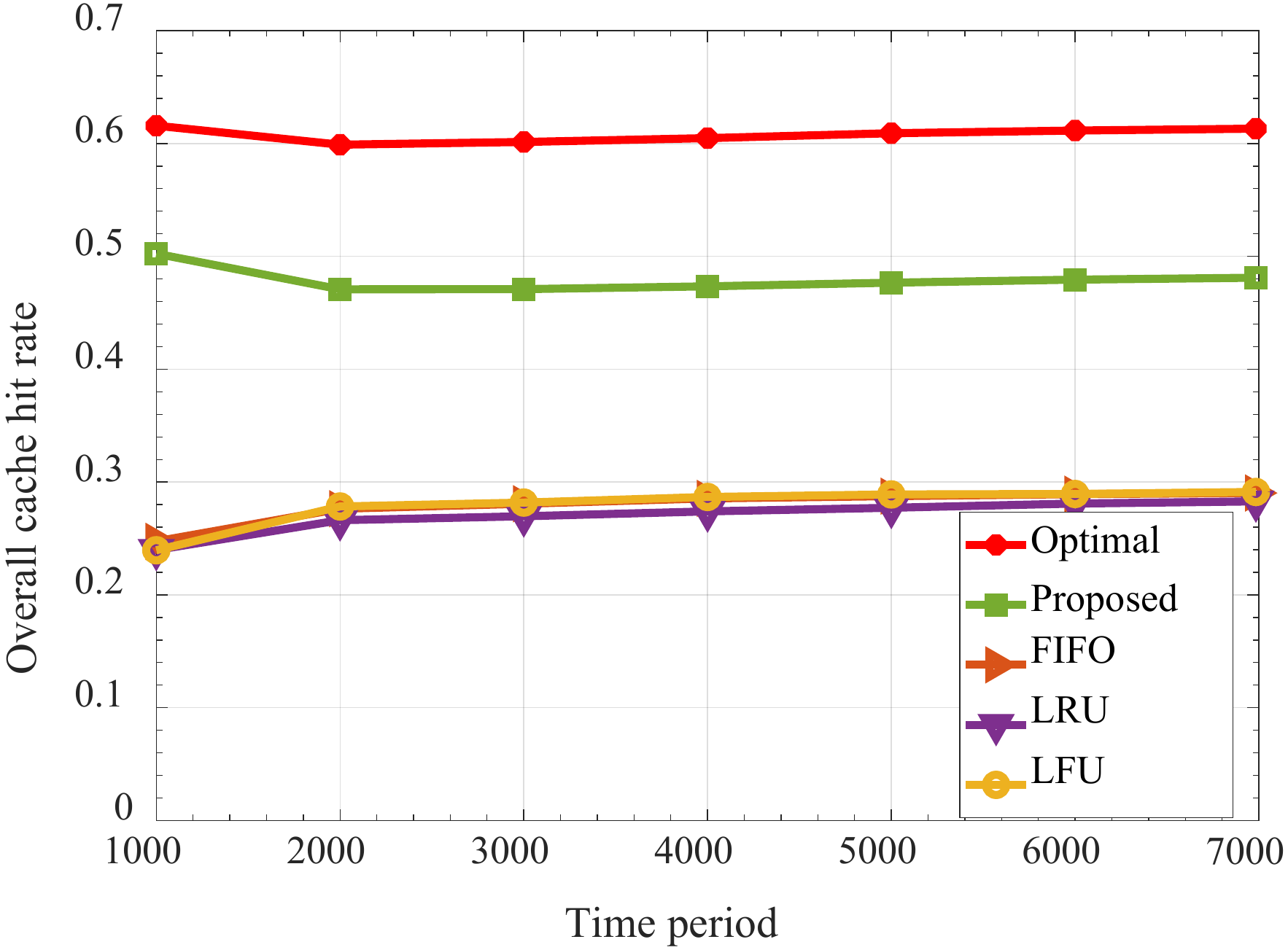}
\caption{\normalsize The overall cache hit rate versus  time period for the proposed  policy and the baseline policies.}
\label{hitTn}
\end{figure}

In Fig. \ref{hitTn}, we show the overall cache hit rates of our proposed policy and the four baseline  policies  until the current time period  with $M \varphi = 1800$.
It can be observed that the overall cache hit rate of our proposed policy follows consistently along with that of the optimal policy.
The reasons are that both of them make caching decisions according to the  content popularity, and that the distributions of the predicted popularity and the real one are consistent during every time period.
It can also be observed that the changes of the overall cache hit rates of all the policies are different. 
The FIFO, LRU and LFU policies have  low overall cache hit rates during the initial time periods due to the inevitable cold-start problem,
whereas our proposed policy can  cache the predicted popular contents according to the already-learned user preference model during the initial time period and then achieve a higher cache hit rate accordingly.
After that, the overall cache hit rates of all the policies gradually increase with the time period.
The reason is that caching decisions can be made more accurate with the increase of  user requests.

\section{Conclusions}

In this paper, we have proposed two edge caching architectures and a novel edge caching policy by learning user preference and predicting content popularity.
Our proposed policy can promptly detect the regional popular contents through online content popularity prediction, and {store it  to the local caches in real time}.
Specifically, we have proposed a self-starting offline user preference model updating mechanism by  monitoring the average logistic loss in real time,
which avoids frequent and blind training.
Analytical results have shown that our proposed policy has the capability of asymptotically approaching the optimal performance. Simulation results have shown that our proposed policy
achieves  a better caching performance (i.e., overall cache hit rate) compared to the other traditional policies.  Future work will explore the idea of on-device caching using distributed and low-latency machine learning.

\appendix
\section*{A. Proof of Theorem \ref{th0}}

It can be readily seen that the objective function of the optimization problem in  \eqref{eqeq-10} is convex.
Therefore, the iteratively updated model parameters can be obtained by setting the first order partial derivative of the corresponding objective function with respect to $\boldsymbol{w}_u$ to  zero as follows
\begin{multline}
{{\partial \left( {{{\left( {{{\boldsymbol{g}}^{\left( {1:k} \right)}}} \right)}^{\rm T}} {{\boldsymbol{w}}_u} + \frac{1}{2}\sum\limits_{k' = 1}^k {{\sigma ^{\left( k' \right)}}\left\| {{{\boldsymbol{w}}_u} - {\boldsymbol{w}}_u^{\left( k' \right)}} \right\|_2^2} } \right)}}/{{\partial {{{\boldsymbol{w}}_u}} }}
\\ =
{{\boldsymbol{g}}^{\left( {1:k} \right)}} + \sum\limits_{k' = 1}^k {{\sigma ^{\left( k' \right)}}\left( {{{\boldsymbol{w}}_u} - {\boldsymbol{w}}_u^{\left( k' \right)}} \right)} = 0.
\end{multline}
Correspondingly, the solution of the optimization problem in \eqref{eqeq-10}, i.e., ${\boldsymbol{w}}_u^{\left( {k + 1} \right)}$, should satisfy the following relationship
\begin{align}
\sum\limits_{k' = 1}^k {{\sigma ^{\left( k' \right)}}{\boldsymbol{w}}_u^{\left( {k + 1} \right)}}  = \sum\limits_{k' = 1}^k {{\sigma ^{\left( k' \right)}}{\boldsymbol{w}}_u^{\left( k' \right)}}  - {{\boldsymbol{g}}^{\left( {1:k} \right)}}.
\end{align}

Utilizing the relationship $\sum\nolimits_{k' = 1}^k {{\sigma ^{\left( k' \right)}}}  = {1}/{{{\eta ^{\left( k \right)}}}}$, we can establish that
\begin{align}\label{eq-14}
\frac{1}{{{\eta ^{\left( k \right)}}}}{\boldsymbol{w}}_u^{\left( {k + 1} \right)} = \sum\limits_{k' = 1}^k {{\sigma ^{\left( k' \right)}}{\boldsymbol{w}}_u^{\left( k' \right)}}  - {{\boldsymbol{g}}^{\left( {1:k} \right)}}.
\end{align}
When $k = 1$, from \eqref{eq-14}, we can readily establish:
$
{\boldsymbol{w}}_u^{\left( 2 \right)} = {\boldsymbol{w}}_u^{\left( 1 \right)} - {\eta ^{\left( 1 \right)}}{{\boldsymbol{g}}^{\left( 1 \right)}}.
$
When $k = 2,3, \cdots, K,$
replace $k$ in \eqref{eq-14} by $k-1$. Then, we have
\begin{align}\label{eq-15}
\frac{1}{{{\eta ^{\left( {k - 1} \right)}}}}{\boldsymbol{w}}_u^{\left( k \right)} = \sum\limits_{k' = 1}^{k - 1} {{\sigma ^{\left( k' \right)}}{\boldsymbol{w}}_u^{\left( k' \right)}}  - {{\boldsymbol{g}}^{\left( {1:k - 1} \right)}}.
\end{align}
From \eqref{eq-14} and \eqref{eq-15}, we can readily establish that
\begin{align}
\frac{1}{{{\eta ^{\left( k \right)}}}}{\boldsymbol{w}}_u^{\left( {k + 1} \right)} - \frac{1}{{{\eta ^{\left( {k - 1} \right)}}}}{\boldsymbol{w}}_u^{\left( k \right)} = {\sigma ^{\left( k \right)}}{\boldsymbol{w}}_u^{\left( k \right)} - {{\boldsymbol{g}}^{\left( k \right)}}.
\end{align}
Exploiting the relationship ${\sigma ^{\left( k \right)}} = \left( {{1}/{{{\eta ^{\left( k \right)}}}} - {1}/{{{\eta ^{\left( {k - 1} \right)}}}}} \right)$, we can further establish that
\begin{equation}
\frac{1}{{{\eta ^{\left( k \right)}}}}{\boldsymbol{w}}_u^{\left( {k + 1} \right)} - \frac{1}{{{\eta ^{\left( {k - 1} \right)}}}}{\boldsymbol{w}}_u^{\left( k \right)}=
 \left( {\frac{1}{{{\eta ^{\left( k \right)}}}} - \frac{1}{{{\eta ^{\left( {k - 1} \right)}}}}} \right){\boldsymbol{w}}_u^{\left( k \right)} - {{\boldsymbol{g}}^{\left( k \right)}}.
\end{equation}
Then, we have
\begin{equation}\label{eq_222}
{\boldsymbol{w}}_u^{(k+1)} = {\boldsymbol{w}}_u^{(k)} - \eta^{(k)} \boldsymbol{g}^{(k)}, \ k=2,3,\cdots,K.
\end{equation}
According to the above analysis, we can certainly establish
\begin{equation}
{\boldsymbol{w}}_u^{(k+1)} = {\boldsymbol{w}}_u^{(k)} - \eta^{(k)} \boldsymbol{g}^{(k)}, \ k=1,2,\cdots,K.
\end{equation}

It is obvious that the above solution of the optimization problem in \eqref{eqeq-10} is the same as the solution of the optimization problem in \eqref{eq-10}, i.e., \eqref{eq-18}. This completes the proof.

\section*{B. Proof of Theorem \ref{th1}}

We first analyze the upper bound of the expected prediction error of the content requested possibility of one single user for the overall $D$ requests in the finite time horizon $T$.  Without loss of generality, {the convex loss function} is chosen to be an absolute loss function, i.e., ${L_d}\left( {{{\boldsymbol{w}}_u}} \right) = \left| {{{\hat p}_{t,u,d}} - {p_{t,u,d}}} \right|$. Then, from \eqref{eq17}, the following relationship can be readily established
\begin{align}\label{eq18}
\mathbb{E} \sum\limits_{t = 1}^T {\sum\limits_{d = 1}^{{D_t}} {\left| {{{\hat p}_{t,u,d}} - {p_{t,u,d}}} \right| \le } } {W_u}{G_u}\sqrt {2{D}}  + {\tau _u}.
\end{align}
Furthermore, by using the relationship in \eqref{eq5}, the expected  popularity prediction error for the overall $D$ requests in the finite time horizon $T$ can be formulated as follows
\begin{equation}
\mathbb{E}\sum\limits_{t = 1}^T {\sum\limits_{d = 1}^{{D_t}} {\left| {{{\hat P}_{t,d}} - {P_{t,d}}} \right|} }  =
\mathbb{E}\sum\limits_{t = 1}^T {\sum\limits_{d = 1}^{{D_t}} {{{\sum\limits_{u = 1}^{{U_t}} \frac{1} {{{U_t}}}{\left| {{{\hat p}_{t,u,d}} - {p_{t,u,d}}} \right|} }}} }.
\end{equation}
By considering that $U_{\min } \le U_t \le U_{\max }$,
the following inequation can be readily established
\begin{equation}
\mathbb{E}\sum\limits_{t = 1}^T {\sum\limits_{d = 1}^{{D_t}} {\left| {{{\hat P}_{t,d}} - {P_{t,d}}} \right|} }  \le  \frac{1}{{{U_{\min }}}}{\mathbb{E}{\sum\limits_{u = 1}^{{U_{\max }}} {\sum\limits_{t = 1}^T {\sum\limits_{d = 1}^{{D_t}} {\left| {{{\hat p}_{t,u,d}} - {p_{t,u,d}}} \right|} } } }}.
\end{equation}
Then, from \eqref{eq18}, we can obtain
\begin{equation}
\mathbb{E}\sum\limits_{t = 1}^T {\sum\limits_{d = 1}^{{D_t}} {\left| {{{\hat P}_{t,d}} - {P_{t,d}}} \right|} }  \le
\frac {U_{\max }}{{{U_{\min }}}} \left( {{W_{\max }}{G_{\max }}\sqrt {2D}  + {\tau _{\max }}} \right).
\end{equation}

This completes the proof.

\section*{C. Proof of Theorem \ref{th2}}

During each time period, the proposed  edge caching policy always tries to cache the ${M\varphi }$ most popular contents.
In the ideal case,  the  contents with the ${M\varphi }$  largest real  popularities, i.e., $\left\{ {P_{t,1}^{'},P_{t,2}^{'},\cdots,P_{t,M\varphi }^{'}} \right\}$, will be cached.
Due to the popularity prediction errors, the contents with the ${M\varphi }$  largest predicted   popularities
will however be cached in our proposed edge caching policy.
Assume that the predicted   popularities are sorted in descending order as follows: $\hat P_{t,{{d_1'}}}^{{'}} \ge \hat P_{t,{{d_2'}}}^{{'}} \ge  \cdots  \ge \hat P_{t,{{d_f'}}}^{{'}} \ge  \cdots  \ge \hat P_{t,{{d_{F_t}'}}}^{{'}}.$
Obviously, $\{ {\hat P_{t,{{d_1'}}}^{'}, \hat P_{t,{{d_2'}}}^{'},\cdots, \hat P_{t,{{d_{M\varphi}'}}}^{'}} \}$ represent the  ${M\varphi }$  largest predicted  popularities.
Then, the following relationship can be readily established
\begin{equation}\label{eq22}
\sum\limits_{f = 1}^{M\varphi } {\hat P_{t,{d_f'}}^{{'}}}  \ge \sum\limits_{d = 1}^{M\varphi } {\hat P_{t,d'}^{{'}}}.
\end{equation}
According to  the definition of $\Delta {P_t}$, we have $| {\hat P_{t,d'}^{{'}} - P_{t,d'}^{{'}}} | \le \Delta {P_t}$.
Correspondingly, we have
\begin{equation}
\hat P_{t,d'}^{{'}}  \ge P_{t,d'}^{{'}} - \Delta {P_t}, \
\sum\limits_{d' = 1}^{M\varphi } {\hat P_{t,d'}^{{'}}} \ge \sum\limits_{d' = 1}^{M\varphi } {\left( {P_{t,d'}^{{'}} - \Delta {P_t}} \right)}. \label{eq31}
\end{equation}
From \eqref{eq20}, the following relationship holds
\begin{equation} \label{eq32}
{{\sum\limits_{d' = 1}^{M\varphi } {P_{t,d'}^{{'}}} }} = \mathcal{H}_t \left( \Phi ^* \right) {{\sum\limits_{d' = 1}^{F_t} {P_{t,d'}^{{'}}} }}.
\end{equation}
Correspondingly, from \eqref{eq22}, \eqref{eq31} and \eqref{eq32}, we can readily obtain
\begin{align}\label{eq24}
\sum\limits_{f = 1}^{M\varphi } {\hat P_{t,{d_f'}}^{{'}}}
\ge \mathcal{H}_t \left( \Phi ^* \right) \sum\limits_{d' = 1}^{F_t} {P_{t,d'}^{{'}}}  - \Delta {P_t} M\varphi .
\end{align}

According to the previous descriptions in Section II, the achievable cache hit rate $\mathcal{H}_t \left( \Phi  \right)$ during the $t$th time period can be defined as follows
\begin{equation}
\mathcal{H}_t \left( \Phi  \right)  = \frac{1}{{{D_t}}} {{\sum\limits_{d = 1}^{{D_t}} {{\theta _{t,d}}\left( {f\left( d \right)} \right)} }}.
\end{equation}
{We have assumed previously that {the requests of the same content are concentrated} in one time period.
Therefore, the contents with the corresponding popularities $\{ {\hat P_{t,{d_1'}}^{'}, \hat P_{t,{d_2'}}^{'},\cdots, \hat P_{t,{d_{M\varphi}'}}^{'}} \}$ will be cached  after the first content request with an initial cache miss   happens during the $t$th time period. Then, we have
\begin{equation}
{{\sum\limits_{d = 1}^{{D_t}} {{\theta _{t,d}}\left( {f\left( d \right)} \right)} }} =  {{\sum\limits_{f = 1}^{M\varphi } {\left\lfloor { {U_t} \hat P_{t,{d_f'}}^{{'}} - 1} \right\rfloor } }},
\end{equation}
where $\lfloor \cdot \rfloor$ denotes the floor operation.
Then, by using the relationships $D_t =  {{\sum\nolimits_{d' = 1}^{F_t} U_t { P_{t,d'}^{{'}}} }}$ and $\lfloor x \rfloor \ge x -1$, we can further establish the following {inequation}
\begin{equation}\label{eq35}
\mathcal{H}_t \left( \Phi  \right)
 \ge \frac{{\sum\nolimits_{f = 1}^{M\varphi } { {  \hat P_{t,{d_f'}}^{{'}} - 2/{U_t} }  } }}{{\sum\nolimits_{d' = 1}^{F_t} { P_{t,d'}^{{'}}} }}.
\end{equation}}
By utilizing \eqref{eq24},
the following relationship can be readily established
\begin{equation}\label{eq37}
\mathcal{H}_t \left( \Phi  \right)
 \ge \mathcal{H}_t \left( \Phi ^* \right)
 -\frac{M\varphi}{D_t} (\Delta {P_t} U_t + 2 ).
\end{equation}

This completes the proof.

\section*{D. Proof of Theorem \ref{th3}}

From \eqref{eq21}, \eqref{eq-new21}, and \eqref{eq39}, we have
\begin{equation}
R\left( D \right) = \mathbb{E}\frac{1}{{\sum\nolimits_{t = 1}^T {{D_t}} }}{{\sum\limits_{t = 1}^T {\left[ {\mathcal H_t \left( {{\Phi ^*}} \right) - \mathcal H_t \left( \Phi  \right)} \right]{D_t}} }}.
\end{equation}
By utilizing the analytical results from Theorem \ref{th2}, the following relationship can be readily established
\begin{equation}
R\left( D \right) \le \mathbb{E}\frac{{M\varphi}}{D}\sum\limits_{t = 1}^T {(\Delta {P_t} U_t + 2 )}.
\end{equation}
According to the definition of  $ \Delta {P_t} $, we have  $ \Delta {P_t}  \le
\sum\nolimits_{d=1}^{D_t} | {\hat P_{t,d}^{{'}} - P_{t,d}^{{'}}} | $.
Exploiting the above relationship and considering that $U_t \le U_{\rm{max}}$, we can further obtain
\begin{equation}
R\left( D \right) \le
\mathbb{E}\frac{ U_{\rm{max}} M\varphi }{D}\sum\limits_{t = 1}^T {\left( \sum\limits_{d=1}^{D_t} | {\hat P_{t,d}^{{'}} - P_{t,d}^{{'}}} |   + \frac{2}{U_{\rm{max}}} \right)}.
\end{equation}
By using the analytical results from Theorem \ref{th1}, the following relationship can be readily established
\begin{multline}
R\left( D \right) \le \\
\frac{ U_{\rm{max}} M\varphi }{D}
\left[\frac {U_{\max }}{{{U_{\min }}}} \left( {{W_{\max }}{G_{\max }}\sqrt {2D}  + {\tau _{\max }}} \right) + \frac{2T}{U_{\rm{max}}}\right].
\end{multline}

This completes the proof.


%

\bibliographystyle{IEEEtran}
\bibliography{manuscript-jrnl}

\end{document}